\documentclass[12pt]{article}
\usepackage{mathtools}
\usepackage{amssymb,amsmath,amsthm}
\usepackage[mathscr]{euscript}
\usepackage{cite}
\usepackage[colorlinks=true,linkcolor=myblue,citecolor=mymagenta]{hyperref}
\usepackage{lmodern}
\usepackage[T1]{fontenc}
\usepackage{microtype}
\usepackage{enumerate}
\usepackage{xcolor}
\usepackage[letterpaper,hmargin=3.7cm,vmargin=3.3cm]{geometry}
\usepackage{setspace}
\makeatletter
\renewcommand{\section}{\@startsection%
{section}%
{1}%
{0em}%
{1.7em}%
{1.2em}%
{\normalfont\large\centering\bfseries}}
\renewcommand{\@seccntformat}[1]%
{\csname the#1\endcsname.\hspace{0.5em}}
\makeatother

\numberwithin{equation}{section}

\definecolor{mymagenta}{RGB}{188,21,108}
\definecolor{myblue}{RGB}{38,52,150}
\newtheorem{theorem}{Theorem}[section]
\newtheorem{proposition}[theorem]{Proposition}
\newtheorem{lemma}[theorem]{Lemma}
\newtheorem{corollary}[theorem]{Corollary}
\theoremstyle{definition}
\newtheorem{definition}[theorem]{Definition}
\newtheorem{remark}[theorem]{Remark}

\newcommand{\ie}{\emph{i.\,e.}}

\newcommand{\cf}{\emph{cf.}}

\newcommand{\abs}[1]{\left|#1\right|}
\newcommand{\norm}[1]{\left\|#1\right\|}
\newcommand{\inner}[2]{\left\langle#1,#2\right\rangle}
\newcommand{\cc}[1]{\overline{#1}}
\newcommand{\reals}{\mathbb{R}}
\newcommand{\nats}{\mathbb{N}}
\newcommand{\complex}{\mathbb{C}}
\newcommand{\eval}[1]{\upharpoonright_{#1}}

\newcommand{\cH}{\mathcal{H}}
\newcommand{\cB}{\mathcal{B}}

\newcommand{\R}{{\mathbb R}}
\newcommand{\C}{{\mathbb C}}

\newcommand{\convergesto}[2]{\xrightarrow[#1\to#2]{}}
\newcommand{\ournewclass}{\mathscr{S}(\mathcal{H})}

\renewcommand\tilde{\widetilde}

\newcommand{\borel}{\mathfrak{X}}

\DeclareMathOperator{\dom}{dom}
\DeclareMathOperator{\ran}{ran}

\DeclareMathOperator{\Span}{span}

\DeclareMathOperator{\supp}{supp}

\DeclareMathOperator{\im}{Im}

\DeclareMathOperator{\spec}{spec}

\DeclareMathOperator{\Sp}{spec}

\newcommand{\defeq}{\mathrel{\mathop:}=}

\begin{document}
\begin{titlepage}
\title{Point mass perturbations of spectral measures
\footnotetext{%
Mathematics Subject Classification(2010):
47B32,  
34L05,  
47B36 
}
\footnotetext{%
Keywords:
spectral measure;
de Branges spaces;
perturbations of measures
}
\hspace{-8mm}
}
\author{
  \textbf{Rafael del Rio\thanks{%
Supported by UNAM-DGAPA-PAPIIT IN110818}}
\\
\small Departamento de F\'{i}sica Matem\'{a}tica\\[-1.6mm]
\small Instituto de Investigaciones en Matem\'aticas Aplicadas y en Sistemas\\[-1.6mm]
\small Universidad Nacional Aut\'onoma de M\'exico\\[-1.6mm]
\small C.P. 04510, Ciudad de M\'exico\\[-1.6mm]
\small \texttt{delrio@iimas.unam.mx}
\\[2mm]
\textbf{Luis O. Silva\thanks{%
Supported by  UNAM-DGAPA-PAPIIT IN110818 and SEP-CONACYT CB-2015 254062
}}
\\
\small Departamento de F\'{i}sica Matem\'{a}tica\\[-1.6mm]
\small Instituto de Investigaciones en Matem\'aticas Aplicadas y en Sistemas\\[-1.6mm]
\small Universidad Nacional Aut\'onoma de M\'exico\\[-1.6mm]
\small C.P. 04510, Ciudad de M\'exico\\[-1.6mm]
\small \texttt{silva@iimas.unam.mx}
\\[2mm]
\textbf{Julio H. Toloza}\thanks{Partially supported by CONICET (Argentina) 
		under grant PIP 11220150100327CO}
\\
\small INMABB\\[-1.6mm]
\small Departamento de Matem\'atica\\[-1.6mm]
\small Universidad Nacional del Sur (UNS) - CONICET\\[-1.6mm]
\small Bah\'ia Blanca, Argentina\\[-1.6mm]
\small \texttt{julio.toloza@uns.edu.ar}}
\date{}
\maketitle
\vspace{-4mm}
\begin{center}
\begin{minipage}{5in}
  \centerline{{\bf Abstract}} \bigskip Using a generalization of the
  moment problem and the extremal properties of spectral measures
  corresponding to the selfadjoint extensions of a regular symmetric
  operator, we study point mass perturbations of spectral measures. We
  obtain general results for a wide class of operators and apply them
  to the analysis of point mass perturbations of spectral measures
  pertaining to Bessel and generalized Schr\"odinger operators.
\end{minipage}
\end{center}
\thispagestyle{empty}
\end{titlepage}
\mathtoolsset{showonlyrefs}
\section{Introduction}
\label{sec:intro}
The motivation of this work stems from the intriguing fact that very
small perturbations of a measure may destroy the density of
polynomials in the space of square integrable functions with respect
to that measure. More precisely: one can find a Borel measure $\mu$ such that
the polynomials are dense in $L_2(\reals,\mu)$, but they are no longer
dense in $L_2(\reals,\widetilde{\mu})$ when $\widetilde{\mu}$ is
obtained by adding a point mass to $\mu$, that is
\begin{equation}
  \label{eq:mu-mu-tilde-intro}
  \widetilde{\mu}=\mu+\delta_\lambda\,,\qquad \lambda\not\in\supp\mu\,,
\end{equation}
where $\delta_\lambda$ denotes a discrete measure which gives weight
only to the point $\lambda$. This instability phenomenon pertains to
the classical Hamburger moment problem and a particular realization
corresponds to the following reformulation. If $\mu$ is an
$N$-extremal solution to an indeterminate moment problem, then
$\widetilde\mu$ is a non-$N$-extremal solution to a moment problem
(see \cite[Thm.\,2.3.2]{MR0184042} and
\cite[Prop.~4.1]{arXiv:1610.01699}). Remarkably, the theory of the
Hamburger moments is the theory of Jacobi operators, \ie\ operators in
$l_2(\nats)$ generated by semi-infinite tridiagonal symmetric matrices
(see Example~\ref{sec:class-moment-probl}). In this context, the
mentioned phenomenon is paraphrased as follows: if $\mu$ is the
spectral measure of a canonical selfadjoint extension of a
nonselfadjoint Jacobi operator, \ie\ a selfadjoint restriction of its
adjoint (see Section~\ref{sec:gener-resolv}), then there is no
canonical selfadjoint extension of any Jacobi operator having
$\widetilde\mu$ as its spectral measure.

This paper is devoted to studying the situation illustrated above in a
setting that allows us to consider not only nonselfadjoint Jacobi
operators, but a wide class of regular symmetric operators with
deficiency indices $(1,1)$, including regular and certain singular
Schr\"odinger operators. Our approach is based on Krein's theory of
symmetric operators \cite{MR0011170,MR0011533,MR0012177}, the theory
of generalized extensions due to Na\u{\i}mark \cite{MR0002714} and a
generalization of the moment problem originally proposed by Livsi\v{c}
\cite{MR0011171}. This theoretical scaffold allows us to prove several
results on point mass perturbations of spectral measures presented in
Section~\ref{sec:mass-point-pert}. The central theorem
(Theorem~\ref{cor:no-self-adjoint-measure}) is general enough to be
applied to the analysis of point mass perturbation of spectral
measures related to generalized Schr\"odinger operators with measures and Bessel
operators. This analysis led to
Theorems~\ref{prop:example-schroedinger-with-measures},
\ref{thm:existence-potential}, and \ref{prop:example-bessel}.

The spectrum of a regular Schr\"odinger operator defined by any
selfadjoint boundary conditions has the following asymptotic behavior
\begin{equation}
  \label{eq:asymptotics-intro}
\lambda_n = c n^2 + O(1),\quad\text{ as } n\to\infty\,,
\end{equation}
(see \cite{MR894477}).  From this formula, by a straightforward
argument (see Remark~\ref{rem:we-note}), one concludes that if $\mu$
is the spectral measure of a regular Schr\"odinger operator given by
selfadjoint boundary conditions, then $\widetilde\mu$, defined as in
\eqref{eq:mu-mu-tilde-intro}, cannot be the spectral measure of any
regular Schr\"odinger operator with any selfadjoint boundary
conditions no matter which potential function is considered. It is
worth remarking that the results of Section~\ref{sec:mass-point-pert}
provide a way for dealing with point mass perturbation of spectral measures for a wide range of
operators even in the case when asymptotics of the kind of
\eqref{eq:asymptotics-intro} are not available. This is illustrated in
Examples~\ref{sec:generalized-operators} and
\ref{sec:bessel-operators}.

Let us elaborate on the theoretical framework developed in this paper
which is of interest in itself irrespective of the results given in
Sections~\ref{sec:gener-moment-probl} and \ref{sec:examples}. The
classical theory of moments tells us that to account for all solutions
to the moment problem, one has to recur to Na\u{\i}mark's theory of
generalized extensions \cite{MR0002713,MR0002714} since the solutions
are the generalized spectral measures (see
Definition~\ref{def:generalized-resolution}) of the corresponding
Jacobi operator. Here we remark that, despite the classicality of the
theory, relatively recent developments \cite{MR2919742,MR3050310} make
use of \cite{MR0002714} to construct solutions to the
classical moment problem with the measure having certain predefined
properties. Our approach makes use not only of Na\u{\i}mark's theory,
but also Krein's theory of symmetric operators with deficiency indices
(1,1) \cite{MR0011170,MR0011533,MR0012177} to study the set of
measures related to an arbitrary regular symmetric operator not
necessarily entire \cite{MR1466698}.  This allows us to characterize
the spectral measures associated with the operator
(Theorem~\ref{thm:formula-rational-p-q}) and describe the extremal
properties of measures pertaining to this operator's selfadjoint
extensions (Theorem~\ref{thm:extremality-of-measures}). The next step
in our method is the use of the functional model for regular symmetric
operators \cite{MR2607264,MR3002855} to map any regular symmetric
operator into the multiplication operator in a de Branges spaces. This
reformulation leads us directly to the generalized moment problem.

There are various ways of generalizing the moment problem. Here, we
have in mind a proposal of generalization that goes back to Livsi\v{c}
\cite{MR0011171} for whom the moment problem consists in finding a
measure $\rho$ such that the inner product in a Hilbert space of
functions is expressed through the inner product in
$L_2(\reals,\rho)$. A related setting of the generalized moment
problem is found in the introduction of \cite{MR1660000}. As pointed
out in \cite[Chap.\,6]{remling-book}, de Branges treated this problem
using a different terminology when the Hilbert space of functions is a
de Branges space \cite{MR0229011}. In this work, we also consider the
generalization of \cite{MR0011171} with the underlying Hilbert space
being a de Branges space, but in contrast to \cite{MR0229011} and
\cite[Chap.\,6]{remling-book}, we focus on the extremal properties of
the solutions to the generalized moment problem.

Let us outline the material of this paper. Section~\ref{sec:prelim}
presents classical results that are included for the sake of
completeness and to introduce the
notation. Section~\ref{sec:extremality-measures} study the extremal
properties of measures associated with regular symmetric operators. A
short survey on de Branges spaces theory and the generalized moment
problem are the subject of Section~\ref{sec:gener-moment-probl}. In
this section, we also present the functional model for regular
symmetric operators by means of which we connect the results on de
Branges spaces to our class of operators. In
Section~\ref{sec:mass-point-pert}, we prove the main results related
to perturbation of spectral measures that results when adding a mass
at a single point. Here we also touch upon the fact that certain
functions which are dense in $L_{2}(\reals,\mu)$ are no longer dense
in $L_{2}(\reals,\mu)$ (see \eqref{eq:mu-mu-tilde-intro}).  These
results may have several applications to the inverse spectral analysis
of operators as is shown in \cite{arXiv:1610.01699} for Jacobi
operators. Finally, in Section~\ref{sec:examples}, we illustrate the
results of this work. The first example is used to put our results in
the context of the classical moment problem. The other two examples
present nontrivial results pertaining differential operators which
include regular and singular Schr\"odinger operators.

\section{Preliminaries}
\label{sec:prelim}

Let $\mathcal{H}$ be a separable Hilbert space. The inner product in
it is denoted by $\inner{\cdot}{\cdot}$ where we agree that the
left-hand-side argument is anti-linear. We denote by
$\oplus$, $\ominus$, and $\dotplus$ the orthogonal sum, the orthogonal
complement, and the direct sum of linear spaces, respectively.

Let $A$ be a not necessarily densely defined linear closed symmetric
operator. If $A$ is not densely defined, then $A^*$ is a multivalued
linear operator (linear relation). In this case, the symmetry
condition, $A\subset A^*$, remains valid (as subspaces of
$\mathcal{H}\oplus\mathcal{H}$). The deficiency indices $n_+(A)$ and
$n_-(A)$ of $A$ are defined as the 
dimension of $\mathcal{H}\ominus\ran(A-z I)$, when
$z\in\complex_+:=\{z\in\complex : \im z > 0\}$ and 
$z\in\complex_-:=\{z\in\complex : \im z < 0\}$, respectively.
If $A$ has equal deficiency indices, then it has canonical selfadjoint 
extensions (that is, selfadjoint restrictions of $A^*$), with some of 
these extensions being proper linear relations if $\dom(A)$ is not dense in 
$\cH$ \cite[Prop.~5.4]{MR1430397}.

What follows is a brief revision of Na\u{\i}mark's theory of selfadjoint 
extensions to a larger space \cite{MR0002714,MR0002713} (see also 
\cite[Appendix 1]{MR1255973} and \cite{MR1660000}), as well as some results
from Krein's theory on regular symmetric operators \cite{MR1466698,MR0011533}.

\subsection{Generalized resolvents}
\label{sec:gener-resolv}

\begin{definition}
  \label{def:generalized-cayley}
  Consider a closed symmetric operator $A$. Let $A^+$ be a
  selfadjoint extension of $A$ in a Hilbert space
  $\cH^+\supset\cH$ ($\cH^+=\cH$ is allowed). For  $w\in\C$ and $z\in\C\setminus\spec(A^+)$, define
  \begin{equation}
  \label{eq:generalized-cayley}
  	V_{A^+}(w,z)\defeq(A^+ -wI)(A^+-zI)^{-1}=
        I + (z-w)R^+(z),
\end{equation}
where $R^+$ is the resolvent of $A^+$ (\cf \cite[Ch.~1 Sec.~2.1]{MR1466698}).
\end{definition}
From the first resolvent identity (see \cite[Ch.~3 Sec.~7]{MR1192782}), 
one verifies that
\begin{equation}
  \label{eq:gen-cayley-properties1}
  V_{A^+}(w,z)=V_{A^+}(z,w)^{-1}\; \text{ and }\;
  V_{A^+}(w,z)V_{A^+}(z,v)=V_{A^+}(w,v)
\end{equation}
hold true for any $v,w,z\in\C\setminus\spec(A^+)$.
Also, it is straightforward to establish that
\begin{equation}
  \label{eq:gen-cayley-props2}
  V_{A^+}(w,z)^*=V_{A^+}(\cc{w},\cc{z}).
\end{equation}
The following assertion arises by combining the first identity in 
\eqref{eq:gen-cayley-properties1} with
\eqref{eq:gen-cayley-props2} (see \cite[Ch.~1 Sec.~2.1]{MR1466698}).

\begin{proposition}
  \label{prop:misc-about-symm-operator2}
  Let $z\in\complex$ and $w\in\complex_+$. For any selfadjoint extension $A^+$ of $A$ such that
  $z\not\in\spec(A^+)$, the operator $V_{A^+}(w,z)$ maps
  $\cH^+\ominus\ran(A-\cc{w}I)$ injectively onto  $\cH^+\ominus\ran(A-\cc{z}I)$.
\end{proposition}

\begin{definition}
\label{def:herglotz}
A scalar function with domain of analyticity containing $\complex_+$ is a Herglotz 
function whenever $\im z>0 \implies \im f(z)\ge 0$. By this
definition, any real constant is a Herglotz function. Denote by $\mathfrak{R}$ the union of the class of Herglotz
functions and the constant $\infty$.
\end{definition}

\begin{definition}
  \label{def:generalized-resolvent}
Denote by $P^+$ the orthogonal projector of $\cH^+$
onto $\cH$ and by $R^+(z)$ the resolvent of $A^+$. The operator
\begin{equation*}
	R(z)\defeq P^+R^+(z)\eval{\cH},\quad z\in\complex\setminus\reals
\end{equation*}
is a generalized resolvent of $A$ \cite{MR0002714} (see also 
\cite[Appendix 1]{MR1255973} and \cite[\S~3]{MR1660000}).
\end{definition}

Fix a canonical selfadjoint extension of a closed symmetric operator
$A$ with deficiency indices $n_+(A)=n_-(A)=1$ and denote it by $A_\infty$.  For
$w\in\complex_+$ and $z\in\complex\setminus\Sp(A_\infty)$, define
\begin{equation}
 \label{eq:psi-definition}
  \psi(z)\defeq V_{A_\infty}(w,z)\psi\,,
\end{equation}
where $\psi\in\ker(A^*-w I)\setminus\{0\}$ 
(\cf \cite[Appendix~1, Sec.~4, Eq.~2']{MR1255973}).

Let $R_\infty(z)$ be resolvent of $A_\infty$. Given
$\tau\in\mathfrak{R}$, define
\begin{equation}
  \label{eq:generalized-resolvents-parametrized}
   R_\tau(z):=R_\infty(z)-\frac{\inner{\psi(\cc{z})}{\cdot}\psi(z)}{\tau(z)+Q(z)}\,,
\end{equation}
where
\begin{equation*}
  Q(z)\defeq i\im w + (z-w)\inner{\psi(\cc{z})}{\psi(w)}
\end{equation*}
is the so-called $Q$-function (see \cite[Appendix 1 Eq.~8]{MR1255973}).
$R_\tau(z)$ is the resolvent of a canonical selfadjoint extension of
$A$ if and only if $\tau(z)\equiv c\in\reals\cup\{\infty\}$. 
Moreover, once a canonical selfadjoint extension $A_\infty$ is fixed, 
\eqref{eq:generalized-resolvents-parametrized} induces a bijection 
between $\mathfrak{R}$ and the set of all generalized resolvents of $A$ 
(see \cite[Thm.~6.2]{MR2486805}, also 
\cite{MR0018341,MR0282239,MR0282238,MR0002714,MR0062954}).
Equation \eqref{eq:generalized-resolvents-parametrized} is known as the Krein-Na\u{\i}mark formula.

\begin{definition}
\label{def:generalized-resolution}
Let $E_{A^+}$ denote the spectral family of $A^+$ given by the
spectral theorem.
Let $\borel$ denote the $\sigma$-algebra of Borel sets on $\R$.
Given a selfadjoint extension $A^+$ of a symmetric operator $A$, the map
\begin{equation*}
F_{A^+}(\partial)
	\defeq P^+E_{A^+}(\partial)\eval{\cH},\quad \partial\in\borel,
\end{equation*}
is called a generalized spectral family of the operator $A$. Furthermore,
given $\phi\in\cH\setminus\{0\}$, we define
\begin{equation}
\label{eq:generalized-spectral-measure}
\sigma_{A^+\!,\phi}(\partial)
\defeq\inner{\phi}{F_{A^+}(\partial)\phi}_{\cH},\quad \partial\in\borel.
\end{equation}
We call $\sigma_{A^+\!,\phi}$ the generalized spectral measure of $A$ associated with $\phi$
and $A^+$.
\end{definition}

Note that 
\begin{equation*}
\sigma_{A^+\!,\phi}(\partial) 
	= \inner{\phi}{E_{A^+}(\partial)\phi}_{\cH^+},\quad \partial\in\borel,
\end{equation*}
so $\sigma_{A^+\!,\phi}$ is also the spectral measure of the selfadjoint 
extension $A^+$ associated with $\phi\in\cH$ and therefore it is a
nonnegative finite scalar measure.
On the basis of \cite[Appendix~I Thm.~2]{MR1255973}, the next result holds true.

\begin{proposition}
  \label{prop:spectral-family-function}
$F$ is a generalized spectral family of $A$ (that is, $F=F_{A^+}$ for some $A^+$) 
if and only if the operator-valued function  $F(t):=F(-\infty,t)$
($t\in\reals$) satisfies, for any $\phi\in\mathcal{H}$, the following conditions:
  \begin{enumerate}[(i)]
  \item The function $\inner{\phi}{F(t)\phi}$ does not decrease when $t$
    increases.\label{1res-identity}
  \item  $F(t)\phi$ is left continuous.
  \item $F(t)\phi\convergesto{t}{-\infty}0$ and
    $F(t)\phi\convergesto{t}{\infty}\phi$.\label{3res-identity}
  \item If $\phi\in\dom(A)$, then
    \begin{equation*}
      \norm{A\phi}^2=\int_\reals t^2d\inner{\phi}{F(t)\phi}\quad\text{ and
      }\quad A\phi=\int_\reals tdF(t)\phi\,.
    \end{equation*}\label{4res-identity}
  \end{enumerate}
\end{proposition}
The next statement is proven on the basis of
Proposition~\ref{prop:spectral-family-function} (see \cite[Appendix 1, pp. 141--143]{MR1255973}).

\begin{proposition}
\label{prop:bijection-herglotz-generalized-resolutions}
There is a bijection between $\mathfrak{R}$ and the generalized spectral
family of any symmetric operator with deficiency indices $(1,1)$. The bijection
is given by the combination of \eqref{eq:generalized-resolvents-parametrized}
and the identity
\begin{equation}
  \label{eq:generalized-resolvent-measure}
  \inner{\phi}{R(z)\phi}_{\cH}
  =\inner{\phi}{R^+(z)\phi}_{\cH^+}=\int_\reals\frac{d\sigma_{A^+\!,\phi}(t)}{t-z}.
\end{equation}
Moreover, the same bijection gives a one-to-one map between 
$\reals\cup\{\infty\}$ and the spectral family corresponding to canonical 
selfadjoint extensions of $A$.
\end{proposition}

\subsection{Regular symmetric operators and gauges}
\label{sec:regular-gauges}

\begin{definition}
\label{def:complete-nonselfadjointness}
A closed symmetric nonselfadjoint operator is said to be completely
nonselfadjoint if it is not a nontrivial orthogonal sum of a symmetric
and a selfadjoint operators.
\end{definition}

\begin{definition}
  \label{def:point-of-regular-type}
  The complex number $z$ is in the set of points of regular
  type of a closed operator $T$ when there exists $c_z>0$ such that
  \begin{equation*}
    \norm{(T-zI)\phi}\ge
    c_z\norm{\phi} \text{ for all }\phi\in\dom(T)\,.
  \end{equation*}
The complement of the set of points of regular type of $T$ is called
the spectral kernel of $T$.
\end{definition}
\begin{definition}
\label{def:regularity}
A closed operator $T$ is regular when the set of points of regular
type is the whole complex plane.
\end{definition}

A regular, closed symmetric operator is
necessarily completely nonselfadjoint since regularity implies that
the spectral kernel is empty and, therefore, the operator cannot have
selfadjoint parts. On the other hand, there are many completely
nonselfadjoint operators that are not regular
\cite[Sec.~2]{zbMATH06526214}. We note that the dimension of
$\ker(A^* -zI)$ is constant for all $z\in\C$ whenever $A$ is regular.

\begin{definition}
  \label{def:ourclass}
 Let $\ournewclass$ denote the set of regular, closed symmetric
 operators in $\cH$ with deficiency indices $n_+(A)=n_-(A)=1$.
\end{definition}

The next assertion is well known for densely defined symmetric
operators \cite{aronszajn,donoghue,MR1466698}. A proof for the case 
under discussion in this work can be found in \cite[Prop.~2.4]{MR3002855}.

\begin{proposition}
\label{prop:properties-of-new-class}
Let $A$ be in $\ournewclass$. The following
statements are true:
\begin{enumerate}[{(i)}]
\item The spectrum of every canonical selfadjoint extension of $A$
	consists solely of isolated eigenvalues of multiplicity one.
\item Every real number is part of the spectrum of one, and only one,
	canonical selfadjoint extension of $A$.
\item The spectra of the canonical selfadjoint extensions of $A$ are
	pairwise interlaced.
\end{enumerate}
\end{proposition}

Note that item (i) above means that every selfadjoint extension of $A$ is
a simple operator \cite[Sec.~69]{MR1255973}. Also, item (iii) above (together with 
the fact that under extensions the spectral kernel does not decrease) 
implies that $A$ is regular if and only if the spectra of any two selfadjoint 
extensions of $A$ do not intersect.

For any $A\in\ournewclass$, let $\kappa\in\cH$ be such that
\begin{equation}
  \label{eq:decomposition-gauge}
  \cH = \ran ( A - z_{0} I ) \dotplus \Span \{ \kappa \},
\end{equation}
for some $z_{0}\in\complex$. For this vector $\kappa$, consider the set
\begin{equation}
  \label{eq:bad-set}
  \{z\in\complex: \kappa\in\ran ( A - z I )\}=
  \{z\in\complex:\mathcal{H}\ne \ran ( A - z I ) \dotplus \Span \{ \kappa \}\}
  \,.
  \end{equation}
The following result is proven in \cite[Thm.~2.2]{MR2607264}. It was first
formulated by Krein without proof in \cite[Thm.~8]{MR0011533}.

\begin{proposition}\label{prop:good-gauge}
  For every $A\in\ournewclass$, there exists $\kappa$ satisfying
  \eqref{eq:decomposition-gauge} such that the corresponding set
  \eqref{eq:bad-set} does not intersect the real line.
\end{proposition}
\begin{definition}
  \label{def:good-gauge}
  We call the vector $\kappa$, whose existence is established in
  Proposition~\ref{prop:good-gauge}, a spectral gauge of $A\in\ournewclass$.
\end{definition}

\begin{lemma}
\label{lem:cyclic-vector}
A spectral gauge of $A\in\ournewclass$ is a generating
  vector (as defined in \cite[Sec. 69]{MR1255973}) for any canonical
  selfadjoint extension of $A$.
\end{lemma}
\begin{proof}
Let $A_\gamma$ be any canonical selfadjoint extension of $A$ and $\kappa$
a spectral gauge of $A$.
Since $\kappa$ is not in $\ran(A-xI)$ for any $x\in\reals$, $\kappa$ has
a nonzero projection onto the one-dimensional space $\ker(A^*-x
I)$. Because $\ker(A^*-x I)=\ker(A_\gamma-x I)$ whenever 
$x\in\spec(A_\gamma)$,
it follows that $\kappa$ has a nonzero projection onto every eigenspace of 
$A_\gamma$ since its spectrum is discrete.
\end{proof}

The following assertion is stated in \cite[Thm.~1]{MR0011533} and
proven in \cite[Thm.~1.2.5]{MR1466698}.
\begin{proposition}
  \label{prop:krein-gorby}
  Let $\kappa$ be a spectral gauge of $A\in\ournewclass$ and $a,b\in\reals$ such
  that $a<b$. If $F$ is a generalized spectral family of $A$, then
  \begin{equation}
    \label{eq:integral-expression-for-family}
    F([a,b])\phi=\int_a^b\frac{\inner{\psi(t)}{\phi}}
    	{\inner{\psi(t)}{\kappa}}dF(t)\kappa
  \end{equation}
  for any $\phi\in\cH$ and $\psi$ given in \eqref{eq:psi-definition}.
\end{proposition}

\section{Extremality of measures}
\label{sec:extremality-measures}

\begin{definition}
\label{def:family-measures}
Consider $A\in\ournewclass$ and let $\kappa$ be a spectral gauge of $A$. Recalling
\eqref{eq:generalized-spectral-measure}, define
\begin{equation*}
\mathcal{V}_\kappa(A) \defeq
\left\{\sigma=\sigma_{A^+\!,\kappa}: A^+ \text{ is selfadjoint extension of } A\right\}.
\end{equation*}
Inside $\mathcal{V}_\kappa(A)$, we single out the set of extremal measures
\begin{equation*}
\mathcal{V}^{\rm ext}_\kappa(A) \defeq
\left\{\sigma=\sigma_{A^+\!,\kappa}: 
	A^+ \text{ is canonical selfadjoint extension of } A\right\}.
\end{equation*}
\end{definition}

The reason for using word ``extremal'' above will be explained 
in Theorem~\ref{thm:extremality-of-measures} below.
In the context of this definition,
Proposition~\ref{prop:bijection-herglotz-generalized-resolutions} implies the following.

\begin{corollary}
  \label{cor:bijection-herglotz-generalized-resolutions}
  For any spectral gauge $\kappa$ of $A\in\ournewclass$, there is a bijection between 
  $\mathfrak{R}$ and
  $\mathcal{V}_\kappa(A)$.  Moreover, this bijection gives a one-to-one mapping
  between $\reals\cup\{\infty\}$ and $\mathcal{V}^{\rm ext}_\kappa(A)$.
\end{corollary}

The bijection referred above depends on the choice of $A_\infty$ (see
Section~\ref{sec:gener-resolv}), but regardless of this choice, all
measures arising from spectral families corresponding to canonical
selfadjoint extensions of $A$ are obtained when $\tau$ runs through
$\reals\cup\{\infty\}$. Note that, as a consequence of
Proposition~\ref{prop:properties-of-new-class}, one has:
\begin{corollary}
  \label{cor:extremal-measures-properties}
  Let $\kappa$ be a spectral gauge of $A\in\ournewclass$. All measures in
  $\mathcal{V}^{\rm ext}_\kappa(A)$ are discrete and their supports are
  pairwise disjoint. For any $x\in\reals$, there is a measure
  $\sigma\in\mathcal{V}^{\rm ext}_\kappa(A)$ such that $\sigma\{x\}>0$.
\end{corollary}

The next statement can be found
in \cite[Chap.~2 Sec.~7.1]{MR1466698} and the first part of it in
\cite[Appendix I~Sec.~4]{MR1255973} (\cf \cite[Thm.~6]{MR0011170}
and \cite{MR0012177}).
\begin{theorem}
  \label{thm:formula-rational-p-q}
  Consider $A\in\ournewclass$. Given $\tau\in\mathfrak{R}$, let
  $R_\tau(z)$ be a generalized resolvent as in
  \eqref{eq:generalized-resolvents-parametrized}. If $\kappa$ is a spectral gauge of
  $A$, then
  \begin{equation*}
     \inner{\kappa}{R_\tau(z)\kappa}=-\frac{\tau(z)A(z)-C(z)}{\tau(z)B(z)-D(z)}\,,
   \end{equation*}
   where
   $A(z),B(z),C(z),D(z)$ are meromorphic functions on $\complex$, analytic on
   $\reals$.
 \end{theorem}
 \begin{proof}
   Reasoning as in \cite[Ch.~42, Sec.~7.1]{MR1466698}, one obtains
   the stated formula by inserting the right-hand side expression 
   from \eqref{eq:generalized-resolvents-parametrized} into
   $\inner{\kappa}{R_\tau(z)\kappa}$ and considering the functions:
   \begin{align}
     \label{eq:a-functions}
     A(z)&\defeq-\frac{\inner{\kappa}{R_\infty(z)\kappa}}{\inner{\psi(\cc{z})}{\kappa}}\\
     \label{eq:b-functions}
     B(z)&\defeq\frac{1}{\inner{\psi(\cc{z})}{\kappa}}\\[1mm] \label{eq:c-functions}
     C(z)&\defeq-Q(z)A(z) - \inner{\kappa}{\psi(z)}\\[3mm] \label{eq:d-functions}
     D(z)&\defeq-Q(z)B(z)\,,
   \end{align}
   where $\psi$ is given in \eqref{eq:psi-definition}.
   One verifies from their
   definition that $A(z),B(z),D(z)$ can be analytically extended to
   any point $z$ outside the set given in \eqref{eq:bad-set}. The
   function $C(z)$ has a priori poles in $\spec(A_\infty)$, but
   by an argument along the lines of \cite[Ch.~2 Secs.~7.3 and 7.4]{MR1466698}, 
   one  concludes that $C(z)$ can be analytically extended to
   any point $z$ outside the set given in \eqref{eq:bad-set}.
 \end{proof}
 \begin{remark}
   \label{rem:simon-akhiezer-functions}
   For Jacobi operators, the functions
   $A(z),B(z),C(z),D(z)$ coincide with the ones given in
   \cite[Eq.~2.28]{MR0184042} and \cite[Thm.~4.9]{MR1627806}. Note that, for the
   Herglotz function $-1/\tau$, \cite[Eq.~4.36]{MR1627806} is indeed an expression
   for $\inner{\kappa}{R_{-1/\tau}(z)\kappa}$.
 \end{remark}

 Reasoning along the lines of the proof of
 \cite[Thm.~4.17]{MR1627806}, one obtains the following assertion from
 Theorem~\ref{thm:formula-rational-p-q}.

\begin{theorem}
  \label{thm:extremality-of-measures}
  Let $\kappa$ be a spectral gauge of $A\in\ournewclass$ and $\lambda$ a real number. If
  $\sigma\in\mathcal{V}_\kappa(A)\setminus\mathcal{V}_\kappa^{\rm ext}(A)$ is such that
  $\sigma\{\lambda\}>0 $, then there exists
  $\widetilde{\sigma}\in\mathcal{V}_\kappa^{\rm ext}(A)$ such that
  $\widetilde{\sigma}\{\lambda\}>\sigma\{\lambda\}$.
\end{theorem}

\section{A generalization of the moment problem}
\label{sec:gener-moment-probl}

There are various ways of generalizing the moment problem. We use the
generalization proposed by Livsi\v{c} \cite{MR0011171} (see also
\cite{MR1660000}). According to \cite{MR0011171}, the generalized
moment problem consist in finding a measure $\rho$ such that the inner
product in a Hilbert space of functions is expressed through the inner
product in $L_2(\reals,\rho)$. Our choice for the Hilbert space of
functions is a de Branges space. In \cite[Chap.\,6]{remling-book}, the
same choice is done for the generalized moment problem and the
treatment is made by means of canonical systems whose theory contains
the one of de Branges spaces (see \cite{roma-mono}). In contrast to
\cite[Chap.\,6]{remling-book}, our final goal is the extremal
properties of measures.

\begin{definition}
\label{def:axiomatic-db}
A Hilbert space of entire functions $\cB$ is a de Branges space if and only
if, for any function $f(z)$ in $\cB$, the following conditions holds:
\begin{enumerate}[({A}1)]
\item For all $w\in\C$, the linear functional
        $f(\cdot)\mapsto f(w)$  is continuous;

\item for every non-real zero $w$ of $f(z)$, the function
        $f(z)(z-\cc{w})(z-w)^{-1}$ belongs to $\cB$
        and has the same norm as $f(z)$;

\item the function $f^\#(z):=\cc{f(\cc{z})}$ also belongs to $\cB$
        and has the same norm as $f(z)$.
\end{enumerate}
\end{definition}
Alternatively, a de Branges space can be defined in terms of an Hermite-Biehler function, 
that is, an entire function $e(z)$ such that $\abs{e(z)}>\abs{e(\cc{z})}$ for all 
$z\in\C_+$ \cite[Ch.~7]{levin}. In that case, the de Branges space is sometimes denoted 
$\cB(e)$ 
and one has 
\begin{equation}
  \label{eq:standard-inner-db}
  \inner{f}{g}_{\mathcal B}=\int_{\reals}\frac{\cc{f(x)}g(x)}{\abs{e(x)}^2} dx\,.
\end{equation}
(see \cite[Sec.~19]{MR0229011}, \cite[Thm.~2.2]{MR1943095},
\cite{zbMATH06526214}). The r.\,h.\,s. of \eqref{eq:standard-inner-db}
is justified by the fact that if $x$ is a real zero
of  $e(z)$, then $x$ is a zero of greater multiplicity for any
function in $\mathcal{B}(e)$. We note that a given de Branges space can be obtained
by different Hermite-Biehler functions \cite[Thm.~1]{MR0133455}.

\begin{definition}\label{def:multiplication-operator}
Given a de Branges space $\cB$, the operator $S$ of multiplication by the independent
variable is defined by
\begin{gather*}
\dom(S) = \{f(z)\in\cB: zf(z)\in\cB\};\quad
(Sf)(z)=zf(z),\quad f(z)\in\dom(S).
\end{gather*}
\end{definition}

It is noteworthy that the operator $S$ is in $\mathscr{S}(\cB)$
\cite[Prop.~4.2, Cor.~4.7]{MR1664343}. Moreover, to any
operator $A$ in the class $\ournewclass$, there corresponds a de
Branges space such that the operator of multiplication in it is
unitarily equivalent to $A$. To elucidate this, we introduced below
the so-called functional model for operators in $\ournewclass$.

An involution (conjugation) $J:\cH\to\cH$ is an antilinear map satisfying
\begin{equation*}
  J^2=I\quad\text{ and }\quad\inner{J\phi}{J\psi}=\inner{\psi}{\phi}
\end{equation*}
for all $\phi,\psi\in\cH$ \cite[Eq.~8.1]{MR566954}.  Fix an operator
$A\in\ournewclass$ and let $J$ be an involution that commutes with the
selfadjoint extensions of $A$ (there is always such involution because
of \cite[Prop.~2.3]{MR3002855}). Consider a function $\xi_A:\C\to\cH$
such that
\begin{enumerate}[({P}1)]
\item $\xi_A(z)$ is zero-free ($\xi_A(z)\ne0$ for all $z\in\complex$) and 
entire,\label{p1}
\item $\xi_A(z)\in\ker(A^*-zI)$ for all $z\in\C$, and
\item $J\xi_A(z)=\xi_A(\cc{z})$ for every $z\in\C$.\label{p3}
\end{enumerate}

Since $A\in\ournewclass$, one has $\dim\ker(A^*-zI)=1$ for all
$z\in\C$. Using this fact, one can prove the following assertion (see
\cite[Prop.~2.12 and Remark 2.13]{MR3002855}).
\begin{lemma}
  \label{lem:xi-up-to-function}
  If $\xi_A^{(1)}:\C\to\cH$ and $\xi_A^{(2)}:\C\to\cH$ are two functions
  satisfying (P1),(P2), and (P3), then there exists a zero-free real
  entire function $g(z)$ such that $\xi_A^{(1)}(z)=g(z)\xi_A^{(2)}(z)$.
\end{lemma}

There is a standard way of constructing a function $\xi_A(z)$ with the desired
properties: Pick a canonical
selfadjoint extension $A_\gamma$ of $A\in\ournewclass$ and let
$h_\gamma(z)$ be a real entire function whose zero set (counting
multiplicities) equals $\spec(A_\gamma)$; the existence of this
function follows from classical theorems on entire functions, see 
\cite[Ch.~7 Sec.~2]{MR0181738}. Then, define
\begin{equation}
  \label{eq:one-def-xi}
  \xi_A(z)\defeq h_\gamma(z)V_{A_\gamma}(w,z)\psi_w\,,
\end{equation}
where $w\in\R\setminus\spec(A_\gamma)$, $\psi_w\in\ker(A^*-wI)\setminus\{0\}$ and 
$V_{A_\gamma}(w,z)$ is given by \eqref{eq:generalized-cayley}. It can be shown 
that \eqref{eq:one-def-xi} obeys (P1), (P2) and (P3), the latter relative to a suitable
involution \cite[Props.~2.3 and 2.11]{MR3002855}.
Note that Lemma~\ref{lem:xi-up-to-function} implies that \eqref{eq:one-def-xi}
does not depend on the choice of the selfadjoint extension $A_\gamma$
nor on $w$ and, furthermore, every function $\xi_A:\C\to\cH$ can be
written in the form of \eqref{eq:one-def-xi}.

Fix $A\in\ournewclass$ and any function $\xi_A:\C\to\cH$ satisfying
(P1), (P2), and (P3). Then define the map
\begin{equation}
\label{eq:defining-phi}
\left(\Phi_A\varphi\right)(z)\defeq\inner{\xi_A(\cc{z})}{\varphi},\quad
	\varphi\in\cH.
\end{equation}
Due to (P1), $\Phi_A$ linearly maps $\cH$ onto a linear manifold
$\Phi_A\cH$ of entire functions. By \cite[Sec.~1]{MR1660000}, the
complete nonselfadjointness condition
\begin{equation*}
  \bigcap_{z\in\complex\setminus\reals}\ran(A-zI)=\{0\}
\end{equation*}
(see \cite[Chap.~1 Thm.~2.1]{MR1466698}) implies that 
$\Phi_A$ is injective \cite{zbMATH06526214}.  Clearly, the linear
space $\Phi_A\cH$ is turned into a Hilbert space by defining
\begin{equation}
  \label{eq:inner-def-functional-model}
  \inner{\Phi_A\eta}{\Phi_A\varphi}\defeq
\inner{\eta}{\varphi}\,.
\end{equation}
The resulting Hilbert space is a de Branges space which is henceforth denoted by
$\mathcal{B}_A$ \cite[Prop.~2.14]{MR3002855}.
\begin{remark}
  \label{rem:unitary-equivalen-A-S}
  If $A\in\ournewclass$, then $\Phi_A A\Phi_A^{-1}$ is the
  multiplication operator $S$ in $\mathcal{B}_A$. For any (canonical) selfadjoint
  extension $A_\gamma$ of $A$,  $\Phi_A A_\gamma\Phi_A^{-1}$ is a (canonical) selfadjoint
  extension of the multiplication operator in  $\mathcal{B}_A$. Moreover,
  if $\kappa$ is a spectral gauge of $A$ and $m(z)=(\Phi_A\kappa)(z)$, it can be proven 
  from
  that
  \begin{equation*}
    \mathcal{V}_\kappa(A)= \mathcal{V}_m(S)\quad\text{ and }\quad
    \mathcal{V}_\kappa^{\rm ext}(A)= \mathcal{V}_m^{\rm ext}(S)\,.
  \end{equation*}
Note that, if $F$ is a generalized spectral family of $A$, then $\Phi_A F\Phi_A^{-1}$
is a generalized spectral family of $S$.
\end{remark}

The reproducing kernel in the de Branges space $\mathcal{B}_A$ is
\[
k(z,w) = \inner{\xi_A(\cc{z})}{\xi_A(\cc{w})}\,.
\]
Note that $k(z,w)$ is anti-entire with respect to its second argument.

\begin{remark}
\label{rem:this-is-a-mess}
It is not difficult to verify that $k(\cdot,w)\in\ker(S^*-\cc{w}I)$ for all $w\in\C$
and moreover
\[
k(z,w) = h(\cc{w})V_{S_\gamma}(\cc{v},\cc{w})k(z,v),
\]
where $S_\gamma$ is a canonical selfadjoint extension of $S$ given in
Definition~\ref{def:multiplication-operator} and $h(w)$ is a real entire
function having zeros at $\spec(S_\gamma)$ (\cf~\eqref{eq:psi-definition} and
\eqref{eq:one-def-xi}).
\end{remark}

\begin{lemma}
  \label{lem:good-gauge-zero-free-in-r}
  If $\kappa$ is a spectral gauge of $A\in\ournewclass$, then the
  corresponding element $m$ in $\mathcal{B}_A$ 
  (given by Remark~\ref{rem:unitary-equivalen-A-S}) has no zeros in
  $\reals$. Moreover, if the function $m$ in the de Branges space $\cB$ is
  such that $m(x)\ne 0$ for
  all $x\in\reals$, then $m$ is a spectral gauge of the operator of
  multiplication $S$ in $\cB$.
\end{lemma}
\begin{proof}
By our definition of spectral gauge,
  \begin{equation*}
    \inner{\xi(\cc{z})}{\kappa}\ne 0
  \end{equation*}
for any $z\in\reals$. The second part of the statement follows from
noticing that $k(\cdot,w)\in\ker(S^*-\cc{w}I)$ for all $w\in\complex$.
\end{proof}
The following definition coincides with
\cite[Def.\,6.1]{remling-book}. We came to it motivated
by a generalization of the moment problem discussed in
\cite{MR0011171} (see also \cite{MR1660000}). Several results in
\cite{MR0229011} pertain to this definition.
\begin{definition}[Generalized moment problem]
  \label{def:generalized-moment-problem}
  Given a de Branges space $\mathcal{B}$, find a Borel measure $\rho$ such that
  \begin{equation*}
    \inner{g}{f}_{\mathcal{B}}=\int_\reals f(x)\cc{g(x)}d\rho(x)
  \end{equation*}
  for every $f,g\in\mathcal{B}$.
\end{definition}
As \eqref{eq:standard-inner-db} shows, the Lebesgue measure
multiplied by $1/\abs{e(\cdot)}^2$ is always a solution to the
generalized moment problem for $\mathcal{B}(e)$. Below, it will be
established that there are other solutions to this generalized moment
problem.

\begin{theorem}
  \label{thm:generalized-moment-problem-operator-measure}
  Let $S$ be the operator of multiplication by the independent
  variable in a de Branges space $\mathcal{B}$ and $m$ be a function in
  $\mathcal{B}$ not vanishing in $\reals$. If $\rho$ is a solution to
  the generalized moment problem for $\mathcal{B}$, then
  \begin{equation*}
  \sigma(\partial)\defeq\int_\partial \abs{m(x)}^2d\rho(x),
  \quad \partial\in\borel,
\end{equation*}
belongs to $\mathcal{V}_m(S)$.
\end{theorem}
\begin{proof}
 Fix $t\in\reals$ and define
 \begin{equation*}
   \Upsilon_t(f,g):=\int_{-\infty}^tf(x)\cc{g(x)}d\rho(x)
 \end{equation*}
which is a sesquilinear form in $\cB\times\cB$. This form is bounded due to
the inequality
\begin{equation*}
  \Upsilon_t(f,f)\le \int_\reals\abs{f(x)}^2d\rho(x)=\norm{f}^2
\end{equation*}
along with the polarization identity. By \cite[Ch.~2 Sec.~4 Thm.~6]{MR1192782} 
and well-known results on sesquilinear forms \cite[Sec.~II.2]{MR751959}, 
there is a bounded operator $F(t)$ defined on the whole space $\mathcal{B}$ such that
\begin{equation*}
  \inner{g}{F(t)f}=\Upsilon_t(f,g)\,.
\end{equation*}
It turns out that $F(t)$ corresponds to a generalized spectral family
in the sense of
Proposition~\ref{prop:spectral-family-function}. Indeed,
\eqref{1res-identity}--\eqref{3res-identity} are verified directly
from the definition. As regards \eqref{4res-identity}, one verifies
\begin{equation*}
  \int_\reals td\inner{g}{F(t)f}=\int_\reals tf(t)\cc{g(t)}d\rho(t)=\inner{g}{Sf}_{\cB}
\end{equation*}
and
\begin{equation*}
   \int_\reals t^2d\inner{f}{F(t)f}=\int_\reals
   t^2f(t)\cc{f(t)}d\rho(t)=\inner{Sf}{Sf}_{\cB}.
 \end{equation*}
 Thus, for any real Borel set $\partial$,
 \begin{equation*}
   \mathcal{V}_m(S)\ni\inner{m}{F(\partial)m}=\int_\partial\abs{m(x)}^2d\rho(x)
 \end{equation*}
 since $m$ is a spectral gauge of $S$ by Lemma~\ref{lem:good-gauge-zero-free-in-r}.
\end{proof}
\begin{lemma}
  \label{lem:gorby-in-dB}
  Let $m\in\cB$ be such that $m(x)\ne 0$ when $x\in\reals$ and $a,b\in\reals$ such
  that $a<b$. If $F$ is a generalized spectral family of $S$, then
    \begin{equation*}
    (F([a,b])f)(z)=\int_a^b\frac{f(t)}{m(t)}dF(t)m(z)
  \end{equation*}
  for any $f\in\cB$.
\end{lemma}
\begin{proof}
According to Remark~\ref{rem:this-is-a-mess}, since $m(z)$ is a spectral gauge
of $S$, the right-hand side of
\eqref{eq:integral-expression-for-family} in this case is
\begin{equation*}
  \int_a^b\frac{\inner{k(\cdot,t)}{f(\cdot)}_{\cB}}{\inner{k(\cdot,t)}{m(\cdot)}_{\cB}}dF(t)m(z)\,.
\end{equation*}
The assertion then follows from Proposition~\ref{prop:krein-gorby}.
\end{proof}

\begin{theorem}
  \label{thm:operator-measure-generalized-moment-problem}
  Let $S$ be the operator of multiplication by the independent
  variable in a de Branges space $\mathcal{B}$ and $m$ be a function in
  $\mathcal{B}$ not vanishing in $\reals$. If
  $\sigma\in\mathcal{V}_m(S)$, then
\begin{equation*}
  \rho(\partial):=\int_\partial \frac{d\sigma(x)}{\abs{m(x)}^2},
  \quad \partial\in\borel,
\end{equation*}
is a solution to the generalized moment problem for $\mathcal{B}$.
\end{theorem}
\begin{proof}
 By the hypothesis and Lemma~\ref{lem:gorby-in-dB},
 \begin{equation*}
   \sigma(\partial)=\inner{m}{F(\partial)m}_{\cB}
 \end{equation*}
for any real Borel set $\partial$. It follows from
Lemma~\ref{lem:gorby-in-dB} that
\begin{equation*}
  \inner{g}{F[a,b]f}_{\cB}=\int_a^b\frac{f(t)}{m(t)}d\inner{g}{F(t)m}_{\cB}
\end{equation*}
Also
\begin{equation*}
  \inner{F(t)g}{m}_{\cB}=\inner{F(a)g}{m}_{\cB}
  + \int_a^t\frac{\cc{g(x)}}{\cc{m(x)}}d\inner{F(x)m}{m}_{\cB}
\end{equation*}
Putting the last two equalities together, one obtains
\begin{equation*}
\inner{g}{F[a,b]f}_\cB
	=\int_a^b\frac{f(t)\cc{g(t)}}{\abs{m(t)}^2}d\sigma(t)=\int_a^bf(t)\cc{g(t)}d\rho(t)
\end{equation*}
For finishing the proof take the limit when $a\to-\infty$ and $b\to+\infty$.
\end{proof}
\begin{lemma}
  \label{lem:multiplication-as-multiplication}
  Let $S$ be the operator of multiplication by the independent
  variable in a de Branges space $\mathcal{B}$ and $m$ be a function in
  $\mathcal{B}$ not vanishing in $\reals$. If $S_\gamma$ is a
  canonical selfadjoint extension of $S$ and
  \begin{equation*}
    \sigma(\partial)\defeq\inner{m}{E_\gamma(\partial)m},
    \quad\partial\in\borel,
  \end{equation*}
  where $E_\gamma$ is the spectral family of $S_\gamma$, then
  \begin{equation*}
     (\mathfrak{f}(S_\gamma)m)(\lambda)=\mathfrak{f}(\lambda)m(\lambda)
   \end{equation*}
 for any
  $\mathfrak{f}$ in $L_2(\reals,\sigma)$ and $\lambda\in\spec(S_\gamma)$.
\end{lemma}
\begin{proof}
  Using the spectral theorem,
  \begin{equation*}
    (\mathfrak{f}(S_\gamma)m)(z)
    	= \left(\int_\reals\mathfrak{f}(t)dE_\gamma(t) m\right)(z)
    	= \sum_{\lambda\in\spec(S_\gamma)}\mathfrak{f}(\lambda)k(z,\lambda)
    \frac{\inner{k(\cdot,\lambda)}{m}_{\mathcal{B}}}{k(\lambda,\lambda)}\,.
  \end{equation*}
From the last expression the result follows due to the properties of the
reproducing kernel.
\end{proof}

\begin{theorem}
  \label{thm:density-of-functions}
  Let $\mathcal{B}$, $S$ and $m$ be as in the previous
  theorem. Assume that $\rho$ is a solution to the generalized moment
  problem and $\sigma$ is related to it as in 
  Theorem~\ref{thm:operator-measure-generalized-moment-problem}.  A
  necessary and sufficient condition for the measure $\sigma$ to be in
  $\mathcal{V}_m^{\rm ext}(S)$ is that the map $\Psi$ between the spaces
  $\mathcal{B}$ and $L_2(\reals,\rho)$  given by
  \begin{equation*}
  \Psi f=f\eval{\reals}
\end{equation*}
  is unitary (that is, linear, surjective and norm-preserving).
\end{theorem}
\begin{proof}
 Since $\rho$ is a solution to
  the generalized moment problem,
  $\Psi\mathcal{B}$ is contained in $L_2(\reals,\rho)$, where
  $\Psi f=f\eval{\reals}$ and $\Psi$
  is isometric. The hypothesis $\{f\eval{\reals} :
  f\in\mathcal{B}\}= L_2(\reals,\rho)$ therefore means that $\Psi$ is
  unitary. Let $A^\rho$ be the operator of multiplication by the
  independent variable in $L_2(\reals,\rho)$ defined in the maximal
  domain. Since $A^\rho$ is selfadjoint, the same holds for
  $\Psi^{-1}A^\rho\Psi$. If $f\in\dom(S)$, then 
  $\Psi f\in\dom(A^\rho)$. Thus, $\Psi^{-1}A^\rho\Psi$ is a canonical
  selfadjoint extension of $S$. Denote $S_\gamma:=\Psi^{-1}A^\rho\Psi$
  and let $E_\gamma$ be the spectral family of $S_\gamma$. For any Borel
  set $\partial\subset\reals$, define $\eta(\partial):=\inner{m}{E_\gamma(\partial)m}$. By
  the canonical map (see \cite[Sec.~69]{MR1255973}),
  \begin{equation*}
    L_2(\reals,\eta)\ni\mathfrak{f}\mapsto\mathfrak{f}(S_\gamma)m\in\mathcal{B}\,.
  \end{equation*}
 and for any function $f$ in $\mathcal{B}$ there is a function
 $\mathfrak{f}$ in $L_2(\reals,\eta)$ such that
 $f=\mathfrak{f}(S_\gamma)m$. Thus
 \begin{equation*}
   \int_\reals f(t)\cc{g(t)}d\rho(t)=\int_\reals \mathfrak{f}(t)\cc{\mathfrak{g}(t)}d\eta(t)\,,
 \end{equation*}
where $f$ and $g$ are the images of $\mathfrak{f}$ and $\mathfrak{g}$
under the canonical map. On the basis of
Lemma~\ref{lem:multiplication-as-multiplication},
$f(t)=\mathfrak{f}(t)m(t)$ and $g(t)=\mathfrak{g}(t)m(t)$ with the
equalities in the $L_2(\reals,\eta)$ sense. Therefore
\begin{equation*}
  \int_\reals f(t)\cc{g(t)}d\rho(t)= \int_\reals f(t)\cc{g(t)}\frac{d\eta(t)}{\abs{m(t)}^2}\,.
\end{equation*}
Since this equality holds for any element in $L_2(\reals,\rho)$, one
concludes that $\eta(\partial)=\sigma(\partial)$ for any real Borel set
$\partial$ and hence $\sigma$ is extremal, that is,
$\sigma\in\mathcal{V}_m^{\rm ext}(S)$.

Let us prove the other direction. If $\sigma\in\mathcal{V}_m^{\rm
  ext}(S)$, then $\sigma$ is the spectral measure of a canonical
selfadjoint extension of $S$ which is known to be simple. By the canonical map (see
\cite[Sec.~69]{MR1255973}), $\mathcal{B}$ is unitarily equivalent to
$L_2(\reals,\sigma)$. In turn, since $m$ is zero-free on the real
axis, $L_2(\reals,\sigma)$  is unitarily equivalent to
$L_2(\reals,\rho)$. On the other hand, the Hilbert spaces $\mathcal{B}$ and
$\Psi\mathcal{B}\subset L_2(\reals,\rho)$ are unitarily equivalent due 
to the fact that $\rho$ is a solution to the moment problem. 
Therefore, $\Psi\mathcal{B}$ cannot be properly contained in $L_2(\reals,\rho)$.
\end{proof}

\begin{definition}
  \label{def:extremal-solution}
  A solution $\rho$ to the generalized moment problem for
  $\mathcal{B}$ is said to be extremal if the map $\Psi:\cB\to L_2(\R,\rho)$ given in
  Theorem~\ref{thm:density-of-functions} is unitary.
\end{definition}


Unlike the extremal measures (see
Definition~\ref{def:family-measures}), the extremal solutions might
not be finite measures. However, the extremal solutions of
the generalized moment problem for $\mathcal{B}$ can be finite measures
if, for instance, $1\in\mathcal{B}$ (see Example~\ref{sec:class-moment-probl}).

Let $\delta_x:\borel\to [0,1]$ be the measure defined by the rule
$\delta_x(\partial)\defeq
	\begin{cases}
	1 & x\in\partial\\
	0 & x\not\in\partial.
\end{cases}$
%

\begin{theorem}
\label{thm:extremal-solution-gen-mp}
Let $\mathcal{B}$ be a de Branges space, $k(\cdot,w)$ its reproducing kernel, and
$S_\gamma$ a canonical selfadjoint extension of the operator of
multiplication $S$.  The measure
\begin{equation}
\label{eq:another-long-shot}
\rho_\gamma:=\sum_{\lambda\in\spec(S_\gamma)}\frac{\delta_\lambda}{k(\lambda,\lambda)}
\end{equation}
is an extremal solution to the generalized moment problem for
$\mathcal{B}$. Reciprocally, every extremal solution is of the form
\eqref{eq:another-long-shot}.
\end{theorem}
\begin{proof}
It follows from Remark~\ref{rem:this-is-a-mess} and the equality 
$\ker(S_\gamma -\lambda I) = \ker(S^* -\lambda I)$
that $\{k(\cdot,\lambda)/\norm{k(\cdot,\lambda)}\}_{\lambda\in\spec(S_\gamma)}$ is an orthonormal basis 
in $\mathcal{B}$. Therefore, on the one hand, a function $f$ is in
$\mathcal{B}$ if and only if
\begin{equation}
  \label{eq:sequence-in-l2}
  \sum_{\lambda\in\spec(S_\gamma)}\frac{\abs{f(\lambda)}^{2}}{k(\lambda,\lambda)}<+\infty
\end{equation}
and, on the other hand, one has the following interpolation formula
\begin{equation}
\label{eq:boring}
f(z) = \sum_{\lambda\in\spec(S_\gamma)} 
		\frac{k(z,\lambda)}{k(\lambda,\lambda)}f(\lambda)
\end{equation}
for every $f\in\cB$, where the convergence is in the sense of Hilbert
space, hence uniform in compact subsets of $\C$.  Thus, if
$\rho_{\gamma}$ is given by \eqref{eq:another-long-shot}, then every
$f$ in $L_2 (\reals,\rho_\gamma)$ satisfies \eqref{eq:sequence-in-l2}
and, through the interpolation formula
\eqref{eq:boring}, it is mapped into a function in $\mathcal{B}$ whose
restriction to the real axis is the function $f$.
The identity
\begin{equation*}
\norm{f}^2_{\mathcal{B}}
 =\sum_{\lambda\in\spec(S_\gamma)}\frac{\abs{f(\lambda)}^2}{k(\lambda,\lambda)}= 
 	\norm{f}^2_{L_2(\reals,\rho_\gamma)}\,.
 \end{equation*}
 implies that $\Psi$ is norm-preserving.

 Conversely, if $\rho$ is an extremal solution, then $\supp\rho$
 coincides with $\spec(S_{\gamma})$ for some $\gamma$ as a consequence
 of Theorem~\ref{thm:density-of-functions}. Thus, for any $f$ in
 $L_2 (\reals,\rho)$ whose zero set is
 $\spec(S_{\gamma})\setminus \{\lambda\}$, one has
 \begin{equation*}
   \frac{\abs{f(\lambda)}^{2}}{k(\lambda,\lambda)}=\norm{f}_{\mathcal{B}}^{2}=
   \norm{f}^2_{L_2(\reals,\rho)}=c\abs{f(\lambda)}^{2}
 \end{equation*}
 which shows that the weight $c$ of the measure $\rho$ at $\lambda$
 should be equal to $1/k(\lambda,\lambda)$.
\end{proof}
\begin{remark}
  \label{rem:measure-in-sl-theory}
  Assume that $A\in\ournewclass$ and take the measure $\rho_\gamma$ as
  defined in Theorem~\ref{thm:extremal-solution-gen-mp} for the
  corresponding de Branges space $\mathcal{B}_A$. The multiplication by the
  independent variable in $L_2(\reals,\rho_\gamma)$ is unitarily
  equivalent to a canonical selfadjoint extension $A_\gamma$ of
  $A$. Note that one can write
  \begin{equation*}
    \rho_\gamma=\sum_{\lambda\in\spec(S_\gamma)}\frac{\delta_\lambda}{\norm{\xi_A(\lambda)}^2}\,.
  \end{equation*}
  It is worth remarking that in the Sturm-Liouville theory (when $A$
  is associated with a regular Sturm-Liouville difference expression)
  the measure $\rho_\gamma$ given above is usually defined via the
  Weyl function and is called the spectral measure of the selfadjoint
  operator $A_\gamma$ (see details in
  Example~\ref{sec:generalized-operators}).
\end{remark}

\section{Point mass perturbations of measures}
\label{sec:mass-point-pert}
\begin{lemma}
  \label{lem:perturbation-db-space}
  Let $\cB$ be a de Branges space with inner product $\inner{\cdot}{\cdot}_\cB$.
  Given $a>0$ and $\lambda\in\reals$, define
  \begin{equation}
    \label{eq:other-inner-product}
    \inner{g}{f}_{\sim} 
    	\defeq \inner{g}{f}_\cB + a\,\cc{g(\lambda)}f(\lambda),\quad
    f,g\in\mathcal{B}.
  \end{equation}
  Then the linear manifold $\cB$ equipped with the inner product 
  $\inner{\cdot}{\cdot}_{\sim}$ is a de Branges space.
\end{lemma}
\begin{proof}
$\cB$ is closed with respect to $\inner{\cdot}{\cdot}_{\sim}$. 
Indeed, let $\{f_n\}_0^\infty\subset\cB$ be a $\sim$-Cauchy sequence. 
Since $\norm{\cdot}_\sim\ge\norm{\cdot}_\cB$, $\{f_n\}_0^\infty$ is also 
$\cB$-Cauchy. Then, there exists $g\in\cB$ such that $\norm{f_n-g}_\cB\to 0$.
Since 
$\abs{f_n(\lambda)-g(\lambda)}^2\le k(\lambda,\lambda)\norm{f_n-g}_\cB^2$,
one obtains
\[
\norm{f_n-g}^2_\sim 
	\le \norm{f_n-g}^2_\cB + a k(\lambda,\lambda)\norm{f_n-g}^2_\cB,
\]
implying the assertion.

Now we prove (A1)--(A3) of Definition~\ref{def:axiomatic-db}. Since
\[
\abs{f(w)-g(w)}^2
	=   \abs{\inner{k(\cdot,w)}{f-g}}^2
	\le k(w,w)\norm{f-g}^2_\cB
	\le k(w,w)\norm{f-g}^2_\sim,
\]
it follows that point evaluation is continuous with respect to the 
$\sim$-norm.

From the equality
\[
\norm{f^\#}^2_\sim
	= \norm{f^\#}^2_\cB + a \abs{f^\#(\lambda)}^2
	= \norm{f}^2_\cB + a \abs{f(\lambda)}^2
	= \norm{f}^2_\sim,	
\]
it follows that the mapping $f\mapsto f^\#$ (see (A3) in 
Definition~\ref{def:axiomatic-db}) is a $\sim$-isometry in $\cB$.

Finally, suppose $w\in\C\setminus\R$ is a zero of $f\in\cB$ and define
$g(z)\defeq(z-\cc{w})(z-w)^{-1}f(z)$. Then, $g\in\cB$ and 
\[
\norm{g}^2_\sim
	= \norm{g}^2_\cB + a \abs{\frac{\lambda-\cc{w}}{\lambda-w}}^2\abs{f(x)}^2
	= \norm{f}^2_\cB + a \abs{f(x)}^2
	= \norm{f}^2_\sim. \qedhere
\]
\end{proof}

Lemma~\ref{lem:perturbation-db-space} has the following corollary.

\begin{corollary}
  \label{cor:perturbed-solution-of-gen-moment-problem}
   Let $\mathcal{B}$ be a de Branges space. If $\rho$ is a solution to the
   generalized moment problem for $\mathcal{B}$, then
  $\rho+a\delta_\lambda$ ($a>0$, $\lambda\not\in\supp\rho$)
is a solution to the generalized moment problem for a de Branges space
$\mathcal{B}_\sim$ having the same elements as $\mathcal{B}$ but with
inner product given by \eqref{eq:other-inner-product}.
\end{corollary}
\begin{proof}
  The statement follows from the equality
  \begin{equation*}
    \int_\reals \cc{g(x)}f(x)d(\rho+a\delta_\lambda)
    =\inner{g}{f}_\cB + a\,\cc{g(\lambda)}f(\lambda).\qedhere
  \end{equation*}
\end{proof}
\begin{theorem}
  \label{thm:generalized-duran}
  Let $S$ be the operator of multiplication by the independent
  variable in a de Branges space $\mathcal{B}$ and $m$ be a function in
  $\mathcal{B}$ not vanishing in $\reals$. If
  $\sigma\in\mathcal{V}_m(S)$, then
\begin{equation*}
  \sigma+a\abs{m(\lambda)}^2\delta_\lambda\,,
  \qquad a>0\,,\quad \lambda\not\in\supp(\sigma)\,,
\end{equation*}
is in $\mathcal{V}_m(\widetilde{S})\setminus\mathcal{V}_m^{\rm ext}(\widetilde{S})$, where
 $\widetilde{S}$ is the multiplication operator in the
 de Branges space $\mathcal{B}_\sim$ given in
 Corollary~\ref{cor:perturbed-solution-of-gen-moment-problem}.
\end{theorem}
\begin{proof}
  By hypothesis, $\rho$ given in
  Theorem~\ref{thm:operator-measure-generalized-moment-problem} is a
  solution to the generalized moment problem for $\mathcal{B}$. By
  Corollary~\ref{cor:perturbed-solution-of-gen-moment-problem}, $\rho+a\delta_\lambda$
  is a solution to the generalized moment problem for
  $\mathcal{B}_\sim$ and Theorem~\ref{thm:generalized-moment-problem-operator-measure}
  implies that
\begin{equation*}
  \sigma+a\abs{m(\lambda)}^2\delta_\lambda\in\mathcal{V}_m(\widetilde{S})\,.
\end{equation*}
On the other hand, Corollary~\ref{cor:extremal-measures-properties} yields
a measure $\sigma_\lambda\in\mathcal{V}_m^{\rm ext}(S)$ such that
$\sigma_\lambda\{\lambda\}>0$. The function $\rho_\lambda$ related to
$\sigma_\lambda$ as in
Theorem~\ref{thm:operator-measure-generalized-moment-problem} is a
solution to the generalized moment problem for $\mathcal{B}$ and
recurring again to Corollary~\ref{cor:perturbed-solution-of-gen-moment-problem}, one
concludes that $\rho_\lambda+ a\delta_\lambda$ is a solution to the 
generalized moment problem for $\mathcal{B}_\sim$. Thus, in view of 
Theorem~\ref{thm:generalized-moment-problem-operator-measure},
\begin{equation*}
  \sigma_\lambda+a\abs{m(\lambda)}^2\delta_\lambda\in\mathcal{V}_m(\widetilde{S})\,.
\end{equation*}
Since
$\lambda\in\supp(\sigma_\lambda)\setminus\supp(\sigma)$, one has
\begin{equation*}
  (\sigma_\lambda+a\abs{m(\lambda)}^2\delta_\lambda)\{\lambda\}>
  (\sigma+a\abs{m(\lambda)}^2\delta_\lambda)\{\lambda\}>0\,.
\end{equation*}
Assume that $\sigma+a\abs{m(\lambda)}^2\delta_\lambda$ is extremal. If
$\sigma_\lambda+a\abs{m(\lambda)}^2\delta_\lambda$ is extremal, then
a contradiction arises from
Corollary~\ref{cor:extremal-measures-properties}. If
$\sigma_\lambda+a\abs{m(\lambda)}^2\delta_\lambda$ is not extremal,
then, using Theorem
\ref{thm:extremality-of-measures}, one obtains an extremal measure giving
weight to $\lambda$ so again we have a contradiction by
Corollary~\ref{cor:extremal-measures-properties}. Therefore
  $\sigma+a\abs{m(\lambda)}^2\delta_\lambda$ is not in
  $\mathcal{V}_m^{\rm ext}(\widetilde{S})$.
\end{proof}
\begin{theorem}
\label{cor:no-self-adjoint-measure}
Let $\mathcal{B}$ and $\mathcal{B}'$ be de Branges spaces such that they
are set-wise equal. If $\rho$ is a solution to the generalized
moment problem for $\mathcal{B}$, then $\rho+a\delta_{\lambda}$
($a>0$, $\lambda\not\in\supp\rho$) is not an extremal solution to the generalized moment
problem for $\mathcal{B}'$.
\end{theorem}
\begin{remark}
  \label{rem:refinements-to-main-assertion}
  Before proving the assertion, there are two points to
  comment. First, the measure $\rho+a\delta_{\lambda}$ is not
  necessarily a solution of the generalized moment problem for $\mathcal{B}'$. Second,
  the particular case when in the hypothesis $\rho$ is an extremal
  solution has interesting applications (see
  Corollary~\ref{cor:destruction-density} and Section~\ref{sec:examples}).
\end{remark}
\begin{proof}
Assume that $\rho+a\delta_{\lambda}$ is an extremal solution to
the generalized moment problem for $\mathcal{B}'$. Thus,
\begin{equation*}
  \norm{f}^{2}_{\mathcal{B'}}=\int_{\reals}\abs{f(x)}^{2}d\rho +
  a\abs{f(\lambda)}^{2}
\end{equation*}
By hypothesis,
\begin{equation*}
  \norm{f}^{2}_{\mathcal{B}}=\int_{\reals}\abs{f(x)}^{2}d\rho\,.
\end{equation*}
Therefore the inner products of $\mathcal{B}$ and $\mathcal{B}'$ are
related as in \eqref{eq:other-inner-product}. Now, since $\rho$ and
$\rho+a\delta_{\lambda}$ generate via
Theorem~\ref{thm:generalized-moment-problem-operator-measure} spectral
measures of the operators of multiplication in $\mathcal{B}$ and $\mathcal{B}'$, one
obtains a contradiction from Theorem~\ref{thm:generalized-duran}.
\end{proof}
\begin{remark}
  \label{rem:poltoratsky}
  There is an alternative proof of
  Theorem~\ref{cor:no-self-adjoint-measure}. Indeed, assume that
  $\rho+ a\delta_{\lambda}$ is an extremal solution of the moment
  problem for $\mathcal{B}'$ and consider a function
  $g\in L_{2}(\reals,\rho+a\delta_{\lambda})$ whose zero set is
  $\supp\rho$. Thus,
  $\norm{g}_{L_{2}(\reals,\rho+a\delta_{\lambda})}\ne 0$ and, since $\Psi$
  is surjective and norm preserving, the function $\Psi^{-1}g$ is a nonzero element of
  $\mathcal{B}'$ and therefore a nonzero element of $\mathcal{B}$. But
  we get a contradiction since the norm of the restriction to the real
  line of $\Psi^{-1}g$ has zero norm in $L_{2}(\reals,\rho)$.
\end{remark}

\begin{corollary}
  \label{cor:destruction-density}
  Let $\rho$ be a solution to the generalized moment problem for a de Branges
  space $\mathcal{B}$. If $\mathcal{A}\subset\Psi\mathcal{B}$ is dense in $L_2(\reals,\rho)$, then $\mathcal{A}$ is not dense in
$L_2(\reals,\rho+a\delta_{\lambda})$, where $a>0$ and $x\not\in\supp(\rho)$.
\end{corollary}
\begin{proof}
  The hypothesis implies that $\rho$ is an extremal solution, \ie\,
  $\Psi\mathcal{B}=L_2(\reals,\rho)$. As in
  Corollary~\ref{cor:perturbed-solution-of-gen-moment-problem}, let
  $\mathcal{B}_{\sim}$ be the de Branges space having the same elements of
  $\mathcal{B}$ with the inner product given by
  \eqref{eq:other-inner-product}. Then $\Psi\mathcal{B}_{\sim}$, which
  contains $\mathcal{A}$, is not dense in
  $L_2(\reals,\rho+a\delta_{\lambda})$ as a consequence of
  Theorem~\ref{cor:no-self-adjoint-measure}.
\end{proof}
\begin{remark}
  \label{rem:l2-zero-class}
  Note that a function is square-integrable with respect to $\rho$ if
  and only if it is square-integrable with respect to
  $\rho+a\delta_{\lambda}$. Nevertheless, there functions in the
  equivalence class of zero in $L_2(\reals,\rho)$ which have nonzero
  norm in $L_2(\reals,\rho)$.
\end{remark}
\section{Examples}
\label{sec:examples}

\subsection{The classical moment problem}
\label{sec:class-moment-probl}

According to \cite[Thm.~2.1.1]{MR0184042}, for a given real sequence
$\{s_k\}_{k=0}^\infty$ , there exists a Borel measure $\rho$ such that
\begin{equation*}
  s_k=\int_\reals t^kd\rho(t)\qquad \text{ for }k=0,1,\dots
\end{equation*}
if and only if, for all $k\in\nats$,
\begin{equation}
  \label{eq:positivity-condition}
  \det
  \begin{pmatrix}
    s_0 & \dots & s_k\\
    \vdots & \vdots & \vdots\\
    s_k & \dots & s_{2k}
  \end{pmatrix}>0\,.
\end{equation}
The measure is said to be a solution to the moment problem given by
the sequence $\{s_k\}_{k=0}^\infty$.
We refer to all sequences $\{s_k\}_{k=0}^\infty$ satisfying
\eqref{eq:positivity-condition} and normalized so that $s_0=1$ as
sequences of moments. To any sequence of moments there corresponds one
and only one Jacobi matrix
\begin{equation}
  \label{eq:jm}
  \begin{pmatrix}
q_1&b_1&0&0&\cdots\\
b_1&q_2&b_2&0&\\
0&b_2&q_3&b_3&\ddots\\
0&0&b_3&q_4&\ddots\\
\vdots&&\ddots&\ddots&\ddots
  \end{pmatrix}\,,
\end{equation}
where $\{q_{k}\}_{k=1}^{\infty}$ is a sequence of real numbers and
$\{b_{k}\}_{k=1}^{\infty}$ is a sequence of positive numbers (see \cite[paragraph after Eq.~1.8]{MR0184042} and
\cite[p.~93]{MR1627806}). Given a separable Hilbert space $\cH$ and
an orthonormal basis $\{\delta_k\}_{k=1}^\infty$, with any matrix
\eqref{eq:jm}, one can uniquely associate a closed symmetric operator
\cite[Sec.~47]{MR1255973}. This operator, denoted by $A$ and called minimal
Jacobi operator, has either deficiency indices $n_+(A)=n_-(A)=0$ or $n_+(A)=n_-(A)=1$
\cite[Ch.~7, Thm.~1.1]{MR0222718}. When the Jacobi operator has
deficiency indices $(0,0)$ the matrix \eqref{eq:jm} is said to be in
the limit point case, otherwise the matrix \eqref{eq:jm} is said to be
in the limit circle case.  Due to the bijection between Jacobi
matrices and sequences of moments, the limit point and limit circle dichotomy
corresponds to the determinate and indeterminate dichotomy for
sequences of moments.

It is established in \cite[Ch.~4]{MR0184042} that
\begin{equation}
  \label{eq:moments-powers-A}
  s_k=\inner{\delta_1}{A^k\delta_1}\qquad \text{ for }k=0,1,\dots
\end{equation}
This shows that if $A$ has deficiency indices $n_+(A)=n_-(A)=1$, then
the spectral measures of canonical selfadjoint extensions are
solutions of the moment problem, \ie\ an indeterminate sequence of
moments admits various solutions. According to
\cite[Cor.\,2.2.4]{MR0184042}, when the deficiency indices $n_+(A)$
and $n_-(A)$ vanish, there are only one solution to the moment problem. This
justifies the terminology.

A Jacobi operator $A$ with deficiency
indices $n_+(A)=n_-(A)=1$ is in $\ournewclass$. Indeed, if one assume that
$\lambda\in\widehat{\spec}(A)$, then, since the spectral kernel does
not decrease under extensions, $\lambda\in\spec(A_\gamma)$ for any
selfadjoint extension $A_\gamma$ of $A$. But the spectra of canonical
selfadjoint extensions of $A$ are disjoint (see the proof of
\cite[Thm.~4.2.4]{MR0184042}.

For a Jacobi operator $A\in\ournewclass$, the function $\xi_A$ defined
in \eqref{eq:one-def-xi} is given by
$\xi_A(z)=\sum_{k=1}^\infty P_{k-1}(z)\delta_k$, where $P_k$ is the k-th
orthogonal polynomial of the first kind associated with \eqref{eq:jm}
(see \cite[Ch.~1 Sec.~2, Ch.~4 Sec.~1]{MR0184042} and
\cite[Ch.~7 Sec.~2 ]{MR0222718}). Thus we have the unitary map
\begin{equation*}
  \Phi_A:\cH\to\mathcal{B}_A\,.
\end{equation*}
Note first that, since the finite linear combinations of the basis
$\{\delta_k\}_{k=1}^\infty$ are dense in $\cH$, the polynomials are
dense in $\mathcal{B}_A$. This is a peculiarity of de Branges spaces generated
by Jacobi operators. Also,
\begin{equation*}
  (\Phi_AA\delta_1)(z)=\inner{\xi_A(\cc{z})}{A\delta_1}
  =\inner{A^*\xi_A(\cc{z})}{\delta_1}
  =\inner{\cc{z}\xi_A(\cc{z})}{\delta_1}=z
\end{equation*}
since $\inner{\xi_A(\cc{z})}{\delta_1}=1$. Thus
\begin{equation}
  \label{eq:phi-operator}
  (\Phi_AA^k\delta_1)(z)=z^k\quad\text{ for all }k\in\nats\cup\{0\}\,.
\end{equation}
\begin{theorem}
  \label{thm:classic-general}
  The measure $\rho$ is a solution to an indeterminate  moment problem if and only if
  $\rho$ is a solution to the generalized moment problem for $\mathcal{B}_A$.
\end{theorem}
\begin{proof}
  Assume that $\rho$ is a solution to the generalized moment problem
  for $\mathcal{B}_A$. Then, in view of \eqref{eq:moments-powers-A}
  and \eqref{eq:phi-operator}, one has
  \begin{equation*}
   s_k=\inner{1}{z^k}_{\mathcal{B}_A}=\int_\reals t^kd\rho(t)\qquad \text{ for }k=0,1,\dots
\end{equation*}
which means that $\rho$ is a solution to the moment problem given by
$\{s_k\}_{k=0}^\infty$.

Now, suppose that $\rho$ is a solution to the moment problem given by
$\{s_k\}_{k=0}^\infty$. Hence, by \eqref{eq:moments-powers-A}, one has
\begin{equation*}
  \int_\reals t^kd\rho(t)=\inner{\delta_1}{A^k\delta_1}\qquad \text{ for }k=0,1,\dots
\end{equation*}
One then verifies that if $R(t)$ is a polynomial, then
\begin{equation*}
  \int_\reals R(t)d\rho(t)=\inner{\delta_1}{R(A)\delta_1}\,.
\end{equation*}
Thus
\begin{equation*}
  \int_\reals \abs{R(t)}^2d\rho(t)=\inner{R(A)\delta_1}{R(A)\delta_1}\,.
\end{equation*}
But, due to \eqref{eq:phi-operator},
\begin{equation*}
  \inner{R(A)\delta_1}{R(A)\delta_1}=\inner{R}{R}_{\mathcal{B}_A}
\end{equation*}
For completing the proof one uses the fact that the polynomials are
dense in $\mathcal{B}_A$ and the polarization identity.
\end{proof}

In the context of the classical moment problem,
Theorem~\ref{thm:generalized-duran} corresponds to 
\cite[Prop.~4.1(a)]{arXiv:1610.01699},
which says that if $\rho$ is an extremal solution to an indeterminate
moment problem, then $\rho+a\delta_{\lambda}$ ($a>0$,
$\lambda\not\in\supp\rho$) is not extremal although is a solutions of
a moment problem. Thus, according to
Corollary~\ref{cor:destruction-density}, the density of the
polynomials in $L_{2}(\reals,\rho)$ no longer holds in
$L_{2}(\reals,\rho+a\delta_{\lambda})$. Note that the sequences
of moments for which $\rho$ and $\rho+a\delta_{\lambda}$ are
different, therefore the corresponding Jacobi operators are different
and the corresponding de Branges spaces are different. However, the associated
de Branges spaces are set-wise equal, which is actually a consequence of the
fact that the set of polynomials are dense in these two spaces, as well as the
argument in the proof of Lemma~\ref{lem:perturbation-db-space}.

We conclude this example by pointing out that it was probably first
mentioned in \cite[Remark 4.5 (iii)]{teschl-suggest} that
the set of points obtained by adding a point to the spectrum of a
selfadjoint extension of a discrete Schr\"odinger operator in the
limit circle case is no longer the spectrum of a
selfadjoint extension of a discrete Schr\"odinger operator in the
limit circle case (cf. \cite[Sec.\,7 Example\,2]{MR0045281} and
\cite[Chap.\,2 Sec.\,7]{MR933088}).

\subsection{Schr\"odinger operators with measures}
\label{sec:generalized-operators}

The following example is based on \cite{benamor} (see also
\cite{luger}).  Let us consider the differential expression given
informally by
\begin{equation*}
\tau_\mu:=-\frac{d^2}{dx^2}+\mu, \quad x\in[0,b],
\end{equation*}
where $\mu$ is a signed Borel measure on $[0,b]$. Properly, for a function $\varphi$
in $AC[0,b]$ we define 
\begin{equation}
\label{eq:quasi-deriv}
\varphi^{[1]}(x) := \varphi'(x) - \int_{[0,x]}\varphi(t)d\mu(t),
\end{equation}
then we define
\begin{equation}
\label{eq:sing-differential expressión}
\begin{gathered}
\dom(\tau_\mu) := \left\{\varphi\in AC[0,b] : \varphi^{[1]}\in AC[0,b]\right\},
\\[1mm]
\tau_\mu\varphi := - (\varphi^{[1]})'.
\end{gathered}
\end{equation}

The derivative of an element $\varphi\in\dom(\tau_\mu)$ has a (unique) representative
for which \eqref{eq:quasi-deriv} holds for every $x\in[0,b]$; in what follows $\varphi'$
will denote this particular representative. From \cite[Thm.~2.4]{benamor} it follows
that $\varphi'$ may have  discontinuities. Indeed,
\begin{equation}
\label{eq:discontinuity-of-deriv}
\varphi'(x) - \varphi'(x-) = \varphi(x)\mu(\{x\}),\quad x\in(0,b].
\end{equation}

Let $\xi(z,x)$ be the solution to the eigenvalue equation $\tau_\mu\varphi = z\varphi$,
$z\in\C$, in the sense given by \eqref{eq:sing-differential expressión}, that
satisfies the initial conditions $\xi(z,0)=1$, $\xi^{[1]}(z,0)=0$.
By \cite[Thm.~2.3]{benamor}, this solution exists and is a real entire
function of $z$ for
every fixed $x\in[0,b]$.

Let $A$ denote the symmetric operator in $L_2(0,b)$ given by
\begin{equation}
  \label{eq:symmetric-operator}
\begin{gathered}
\dom(A) 
	:= 	\left\{\begin{gathered}\varphi\in \dom(\tau_\mu) : \tau_\mu\varphi\in L_2(0,b),
		\\ 
		\varphi^{[1]}(0) = \varphi'(b-) = \varphi(b) = 0\end{gathered}\right\},
\\[1mm]
A\varphi := -(\varphi^{[1]})'.
\end{gathered}
\end{equation}
The boundary condition $\varphi^{[1]}(0) = 0$ is just the usual one
$\varphi'(0)  + h \varphi(0) = 0$ with $h= -\mu(\{0\})$. On the other hand, the
boundary condition at $b$ is consequence of \eqref{eq:discontinuity-of-deriv} plus
the fact that $A$ is the closure of the minimal operator.

Standard arguments yields that $A$ has deficiency indices
$(1,1)$ and $\xi(z,\cdot)\in\ker(A^*-zI)$
for all $z\in\C$. Since Green's identity holds true for this kind of
operators \cite[Thm.~2.2]{benamor}, its canonical selfadjoint
extensions are defined in the usual way,
\begin{gather*}
\dom(A_\gamma) 
	:= 	\left\{\begin{gathered}\varphi\in \dom(\tau_\mu) : 
		\tau_\mu\varphi\in L_2(0,b),\;\varphi^{[1]}(0)=0,
		\\
		\varphi(b)\cos{\gamma}+\varphi'(b)\sin{\gamma}=0,
		\text{ fixed }\gamma\in[0,\pi)
		\end{gathered}\right\},
\\[1mm]
A_\gamma\varphi := -(\varphi^{[1]})'.
\end{gather*}

The de Branges space associated to $A$ is given by
\begin{equation}
\label{eq:de-branges-space-generalized-SL}
\cB_A := \left\{f(z)=\int_0^b\xi(z,x)\varphi(x)dx: \varphi\in L_2(0,b)\right\},
\quad
\norm{f}_{\cB_A} := \norm{\varphi}_{L_2(0,b)}.
\end{equation}
$\cB_A$ is isometrically equal to $\cB(e_b)$, where $e_b(z) = \xi(z,b) + i\xi'(z,b)$
is an Hermite-Biehler function as shown in
\cite[Prop.~4.1]{benamor}. Thus, since $A$ is unitarily equivalent to
the multiplication operator in $\mathcal{B}_{A}$,  $A\in\mathscr{S}(L_{2}(0,b))$.

The solutions of the generalized moment problem for $\mathcal{B}_A$
correspond to the generalized spectral measures of $A$. Let us
consider here the extremal solutions of the generalized
moment problem. According to Theorem~\ref{thm:extremal-solution-gen-mp} and
Remark~\ref{rem:measure-in-sl-theory}, any extremal solution is given by
\begin{equation}
\label{eq:s-l-spectral-measure}
  \rho_\gamma=\sum_{\lambda\in\spec(A_\gamma)}
\frac{\delta_\lambda}{\norm{\xi(\cdot,\lambda)}^2_{L_2(0,b)}}\,,
\end{equation}
where $A_\gamma$ is a selfadjoint extension of $A$. As in the classical
Sturm-Liouville theory, the Fourier transform introduced in 
\eqref{eq:de-branges-space-generalized-SL} generates an unitary map between 
$L_2(0,b)$ and $L_2(\R,d\rho_\gamma)$ \cite[Sec. 3]{benamor}.

\begin{theorem}
  \label{prop:example-schroedinger-with-measures}
  Let $A$ be the generalized Schr\"odinger operator defined by $\tau_\mu$ as in
  \eqref{eq:symmetric-operator}. Let
  $A_\gamma$ an arbitrarily chosen selfadjoint extension of $A$ and
  $\rho_\gamma$ the corresponding spectral measure given in
  \eqref{eq:s-l-spectral-measure}. If
  \begin{equation*}
    \widetilde{\rho}=\rho_\gamma+s\delta_\lambda
  \end{equation*}
  with $\lambda\not\in\spec(A_\gamma)$ and $s>0$, then:
  \begin{enumerate}[(a)]
  \item There is no signed Borel measure $\nu$ in [0,b] such that
    $\widetilde{\rho}$ is the spectral measure corresponding to
    selfadjoint extensions of $\widetilde{A}$ generated by
    $\tau_{\nu}$.
  \item The set of entire functions
    \begin{equation}
      \label{eq:restricted-db-set}
      \left\{f(z)
      	:=\int_0^b\cos(\sqrt{z} x)\varphi(x)dx : \varphi\in L_2(0,b)\right\}
    \end{equation}
 is a proper subset of $L_2(\reals,\widetilde{\rho})$.
  \end{enumerate}
\end{theorem}
\begin{proof}
According to \cite[Thm.~4.4]{benamor}, the set of functions in the de Branges space
$\cB_A$ coincides with the linear set
\begin{equation*}
\left\{f(z) =\int_0^b\cos(\sqrt{z} x)\varphi(x)dx : \varphi\in L_2(0,b)\right\},
\end{equation*}
for any signed Borel measure supported in $[0,b]$. Then, assertion (a) is a 
direct consequence of Theorem~\ref{cor:no-self-adjoint-measure} while (b) follows 
from Theorems~\ref{thm:density-of-functions} and \ref{thm:generalized-duran}.
\end{proof}

\begin{remark}
  \label{rem:we-note}
  We note that regular Schr\"odinger operators are included in
  Theorem~\ref{prop:example-schroedinger-with-measures} as they
  correspond to signed Borel measures $\mu$ that are absolutely
  continuous with respect to the Lebesgue measure, i.e.,
\begin{equation}
  \label{eq:measure-actually-a-potential}
  d\mu(x) = q(x)dx\quad\text{ with } q\in L_1(0,b)\,.
\end{equation}
For these cases, assertion (a) of
Theorem~\ref{prop:example-schroedinger-with-measures} can be also
shown using an elementary argument based on the fact that, the
eigenvalues of a regular Schr\"odinger operator with separated
boundary conditions (Dirichlet case excluded) obey an asymptotic
formula of the form
\begin{equation}
\label{eq:asymptotic-regular-eigenvalues}
\lambda_n = c n^2 + O(1), \quad n\to\infty,
\end{equation}
for some $c>0$. Clearly, the addition of a point to the spectrum
amounts to shifting the enumeration by 1, producing an asymptotic
formula with an additional non trivial term linear in $n$ not present
in \eqref{eq:asymptotic-regular-eigenvalues}.  An analogous reasoning
holds for the Dirichlet case. We remark that
Theorem~\ref{prop:example-schroedinger-with-measures} holds for a wide
class of perturbations of the regular Laplacian for which we do not
have at our disposal an asymptotic formula of the kind of
\eqref{eq:asymptotic-regular-eigenvalues} for the spectra of the
corresponding selfadjoint extensions.
\end{remark}

The example presented here is also related via
\cite[Sec.\,17]{MR1943095} to the Gelfand-Levitan
theory for half-line Schr\"odinger operators. Let us discuss this in
detail. In \cite{MR0045281} and \cite[Sec.\,17]{MR1943095} the spectral measure of the
half-line Schr\"odinger operator is defined as follows. Suppose in
\eqref{eq:measure-actually-a-potential} that
$q\in L_{1,\text{loc}}([0,\infty))$ (with fixed boundary conditions at 0)
and take into account that $\xi(x,z)$ is the
same for any right endpoint $b\in\reals$. A spectral measure of the
half-line Schr\"odinger operator is a measure $\rho$ such that
the Parseval identity
\begin{equation}
  \label{eq:definition-spectral-measure-gl}
  \norm{\varphi}_{L_{2}(0,b)}^{2}=\norm{f}_{\mathcal{B}_{A}}^{2} 
  =\int_{\reals}\abs{f(\lambda)}^{2}d\rho(\lambda)
\end{equation}
holds for every $f=\Phi_{A}\varphi\in\bigcup_{b>0}\mathcal{B}_{A}$ (see
\eqref{eq:de-branges-space-generalized-SL} and compare with
\cite[Eq.\,2]{MR0045281} and \cite[Thm.\,3.2]{MR1943095}).
The classical Weyl theory tells us that a Schr\"odinger expression has
exactly one spectral measure if $q$ is in the limit point case at infinity,
otherwise it has an infinite set of spectral measures. We remark that in the limit circle
case, the definition of spectral measure given above includes measures
which do not correspond to selfadjoint operators in $L_{2}(0,b)$ (see
in \cite[Sec.\,8]{MR0045281} complementary conditions for a spectral
measure to correspond to selfadjoint extensions).

The central result in Gelfand-Levitan theory is a theorem on 
necessary and sufficient conditions for a measure to be a spectral
measure of a half-line Schr\"odinger operator, that is, the necesssary
and sufficient conditions on $\rho$ such that there is a potential
function $q\in L_{1,\text{loc}}([0,\infty))$ ensuring that
\eqref{eq:definition-spectral-measure-gl} holds. Due to the
nonlocality of the conditions (see \cite[Introduction]{MR0045281}),
there is a potential such that a point mass perturbation of a spectral
measure is the spectral measure of the half-line Schr\"odinger operator with
that potential and same boundary condition at 0
(see \cite[Sec.\,7 Example\,2]{MR0045281} and \cite[Chap.\,2 Sec.\,7]{MR933088}).

We emphasize that the last statement does not contradict
Theorem~\ref{prop:example-schroedinger-with-measures}
since the definition of spectral measure in the Gelfand-Levitan theory is different than
the one considered in this work. We believe that a connection between these different 
notions of a spectral measure can be established by resorting to the theory
of chains of de Branges spaces developed in \cite{woracek}. This remains to be done.

To conclude this example, let us revisit the assertion
of Theorem~\ref{prop:example-schroedinger-with-measures}(a) at the
light of \cite[Secs. 5 and 6]{benamor}. 
Since now it is relevant to keep track of the signed measure defining
the operators, let us denote by $A(\mu)$ the operator given in 
\eqref{eq:symmetric-operator}.
By \cite[Thms.\,5.2, 6.1 and 6.2]{benamor}, we know that there is a
bijective correspondence between the set of signed Borel measures $\mu$ in $[o,b]$,
and the set of even functions $v$ on $(-2b,2b)$ of bounded variation such that
$v(0)=-\mu(\{0\})$ and the operator $I+K_{v}$ is positive definite, where
\begin{equation*}
(K_{v}\varphi)(x)=\frac12\int_{0}^{b}\left(v(t-x)+v(t+x)\right)\varphi(t)dt.
\end{equation*}
By this correspondence, one has
\begin{equation*}
\norm{f}_{\mathcal{B}_{A(\nu)}}^{2}
	=\inner{\Phi_{A(0)}^{-1}f}{(I+K_{v})\Phi_{A(0)}^{-1}f}_{L_{2}(0,b)}
\end{equation*}

\begin{theorem}
  \label{thm:existence-potential}
  Assume the same hypothesis of
  Theorem~\ref{prop:example-schroedinger-with-measures}, then there is
  a signed Borel measure $\nu$ such that $\widetilde\rho$ is a 
  non extremal solution to the
  generalized moment problem for $\mathcal{B}_{A(\nu)}$.
\end{theorem}
\begin{proof}
In view of Theorem~\ref{prop:example-schroedinger-with-measures}(a),
it only remains to prove the existence of a signed Borel measure
$\nu$ such that, for the corresponding space
$\mathcal{B}_{A(\nu)}$, the equality
\begin{equation}
\label{eq:parseval-identity-to-prove}
\norm{f}_{\mathcal{B}_{A(\nu)}}^{2}
	=\int_{\reals}\abs{f(\lambda)}^{2}d\widetilde\rho(\lambda).
\end{equation}
holds. From the linearity of $K_v$ as a map on the function $v$, 
we can assume without loss of generality that $\rho_{\gamma}$ 
is given by \eqref{eq:s-l-spectral-measure} for a selfadjoint extension of
$A(0)$. Thus,
  \begin{equation*}
\norm{f}_{\mathcal{B}_{A(\nu)}}^{2}
	=\norm{\Phi_{A(0)}^{-1}f}_{L_{2}(0,b)}^{2} + s\abs{f(\lambda)}^{2}\,.
\end{equation*}
According to \cite[Thm.\, 6.1]{benamor} (see above), the proof will be
established, once we find an even function $v$ on $(-2b,2b)$ of bounded 
variation such that
\begin{equation*}
\inner{\Phi_{A(0)}^{-1}f}{K_{v}\Phi_{A(0)}^{-1}f}
	=s\abs{f(\lambda)}^{2}.
\end{equation*}
However, using the notation introduced in Section \ref{sec:gener-moment-probl},
\begin{align*}
s\abs{f(\lambda)}^{2}
&=s\inner{\Phi_{A(0)}^{-1}f}{\xi_{A(0)}(\lambda)}
	\inner{\xi_{A(0)}(\lambda)}{\Phi_{A(0)}^{-1}f}
\\
&=\inner{\Phi_{A(0)}^{-1}f}{s\inner{\xi_{A(0)}(\lambda)}{\Phi_{A(0)}^{-1}f}
	\xi_{A(0)}}\,.
\end{align*}
So $K_{v}$ must be equal to
$s\inner{\xi_{A(0)}(\lambda)}{\cdot}\xi_{A(0)}(\lambda)$. Therefore,
using the fact that $\xi_{A(0)}(\lambda,\cdot)=\cos(\sqrt{\lambda}\,\cdot)$, we have
\begin{align*}
  (K_{v}\varphi)(x)&=s\int_{0}^{b}\xi_{A(0)}(\lambda,t)\xi_{A(0)}(\lambda,x)\varphi(t)dt\\
  &=
    s\int_{0}^{b}\cos(\sqrt{\lambda}t)\cos(\sqrt{\lambda}x)\varphi(t)dt\\
  &=\frac{s}{2}\int_{0}^{b}\left(\cos(\sqrt{\lambda}(t+s))+\cos(\sqrt{\lambda}(t-s))\right)\varphi(t)dt\,,
\end{align*}
which yields $v(\cdot)=s\cos(\sqrt{\lambda}\,\cdot)$.
\end{proof}
Note that, since $\nu(\{0\})=v(0)$, the point mass perturbation of a 
extremal spectral measure of a regular Schr\"odinger operator is a
(non extremal) spectral measure of a Schr\"odinger operator 
with measure-valued potential $\nu$ such that $\nu(\{0\})\ne 0$.

\subsection{Bessel operators}
\label{sec:bessel-operators}
Given $b\in(0,\infty)$, consider the differential expression
\begin{equation}
	\label{eq:differential-expression}
	\tau_q \defeq -\frac{d^2}{dx^2}+
			\frac{\nu^2-1/4}{x^2}+q(x), \quad x\in(0,b), 
			\quad \nu \in [0,\infty).
\end{equation}
We assume that $q\in L_{1,{\text{loc}}}(0,b)$ is a real-valued function such that
$\widetilde{q}\in L_1(0,b)$, where 
\begin{equation}
\label{eq:conditions-q}
\widetilde{q}(x)
	\defeq \begin{cases}
			x q(x) 						 & \text{if } \nu > 0,\\
			x \left(1-\log(x)\right)q(x) & \text{if } \nu = 0.
\end{cases}
\end{equation}

As shown in \cite[Thm. 2.4]{kostenko}, $\tau_q$ is regular at $x=b$ whereas at $x=0$ 
it is in the limit point case if $\nu\geq 1$ or in the limit circle case 
if $\nu\in [0,1)$. For the later case we assume the additional boundary condition
\begin{equation}
\label{eq:boundary-left-condition-limit-circle}
\lim_{x \to 0+}x^{\nu-\frac12}\left((\nu+1/2)\varphi(x)-x\varphi'(x)\right) = 0.
\end{equation}

The expression \eqref{eq:differential-expression}, along with the boundary 
condition \eqref{eq:boundary-left-condition-limit-circle} when $\nu\in [0,1)$,
gives rise to a closed, regular, symmetric operator $A$, whose deficiency indices
$(1,1)$ \cite[Sect. 4]{siltol2}.

The corresponding canonical selfadjoint extensions $A_{\gamma}$ are defined as usual,
\begin{equation}
\label{eq:bessel-self-adjoint}
\begin{gathered}
\dom(A_{\gamma})
	\defeq\left\{\begin{gathered}
			\varphi\in L^2(0,b): \varphi,\varphi'\in\text{AC}(0,b],\;
			\tau_q\varphi\in L^2(0,b),
			\\
			\text{ boundary condition }
			\eqref{eq:boundary-left-condition-limit-circle}
			\text{ if }\nu\in [0,1),
			\\
			\varphi(b)\cos{\gamma}+\varphi'(b)\sin{\gamma} = 0,
			\text{ fixed }\gamma\in[0,\pi)
			\end{gathered}\right\},
	\\[1mm]
	A_\gamma\varphi \defeq \tau_q\varphi.
\end{gathered}
\end{equation}
For every $\gamma\in[0,\pi)$, the spectrum of $A_{\gamma}$, beside of being simple
and discrete, has at most a finite number of negative eigenvalues 
\cite[Thm. 2.4]{kostenko}.

By \cite[Lemma 2.2]{kostenko}, the eigenvalue equation 
$\tau\varphi = z\varphi$ ($z\in\C$) admits a solution $\xi(z,x)$, 
real entire with respect to $z$, with derivative $\xi'(z,x)$ also
real entire. Moreover, $\xi(z,x)$ also obeys boundary condition
\eqref{eq:boundary-left-condition-limit-circle} whenever $\nu\in(0,1)$.
This in turn implies $\xi(z,\cdot) \in \ker(A^* - z I)$ for all $z\in\C$.
Hence, the associated de Branges space is
\begin{equation}
\label{eq:associated-dB-space}
\cB_{A} 
	:= \left\{f(z)=\int_0^b \xi(z,x)\varphi(x)dx : \varphi \in 
	L_2(0,s)\right\},\quad
\norm{F}_{\cB_{A}} 
		= \norm{\varphi}_{L_2(0,b)}.
\end{equation}
Also, the extremal solutions of the corresponding generalized moment problem are
given by
\begin{equation}
\label{eq:bessel-spectral-measure}
  \rho_\gamma=\sum_{\lambda\in\spec(A_\gamma)}
\frac{\delta_\lambda}{\norm{\xi(\cdot,\lambda)}^2_{L_2(0,b)}},
\end{equation}
where $A_\gamma$ is any selfadjoint extension of $A$. We now have the following
statement analogous to Theorem~\ref{prop:example-schroedinger-with-measures}.

\begin{theorem}
\label{prop:example-bessel}
Suppose $\nu>0$, and $q\in L_{1,{\text{loc}}}(0,b)$ such that $\widetilde{q}\in L_r(0,b)$ 
for some $r\in(2,\infty]$.
Let $A$ be the Bessel operator defined by $\tau_q$ as in
\eqref{eq:symmetric-operator}. Let $A_\gamma$ an arbitrarily chosen selfadjoint 
extension of $A$ and $\rho_\gamma$ the corresponding spectral measure given in
\eqref{eq:bessel-spectral-measure}. If
\begin{equation*}
	\widetilde{\rho}=\rho_\gamma+s\delta_\lambda
\end{equation*}
with $\lambda\not\in\spec(A_\gamma)$ and $s>0$, then:
\begin{enumerate}[(a)]
\item There is no $p\in L_{1,{\text{loc}}}(0,b)$ with
	$\widetilde{p}\in L_r(0,b)$ for some $r\in(2,\infty]$
	such that
	$\widetilde{\rho}$ is the spectral measure corresponding to
	selfadjoint extensions of $\widetilde{A}$ generated by
	$\tau_{p}$.
\item The set of entire functions
	\begin{equation}
	\label{eq:zero-db-set}
	\left\{f(z)\defeq
		\sqrt{\frac{\pi}{2}} z^{-\frac{\nu}{2}}
		\int_0^b\sqrt{x}J_\nu(\sqrt{z} x)\varphi(x)dx : \varphi\in L_2(0,b)\right\}
	\end{equation}
	is a proper subset of $L_2(\reals,\widetilde{\rho})$, where
	$J_\nu$ denotes the Bessel function of the first kind.
\end{enumerate}
\end{theorem}
\begin{proof}
The assertions follow from Theorem~\ref{thm:density-of-functions}, 
Theorem~\ref{thm:generalized-duran} and Corollary~\ref{cor:no-self-adjoint-measure},
this time combined with Theorem~4.2 of \cite{siltol2}.
\end{proof}

\begin{remark}
The requirement on $\tilde{q}$ stated in Theorem~\ref{prop:example-bessel} 
is a technical limitation related to the perturbative argument used in the proof
of \cite[Thm.~4.2]{siltol2}. The technique used there can be easily modified to 
include $r=2$ but breaks down for $r\in[1,2)$.
\end{remark}

We emphasize that our theory does not rule out the possibility that $\tilde{\rho}$ is 
the spectral measure of a Bessel operator for some different value of the 
parameter $\nu$. In this respect, it is worth noticing the following:
Assuming $r\ge 2$ and resorting to \cite[Thm.~2.5]{kostenko}, 
one can show that the eigenvalues of $A_\gamma$ obeys the asymptotic formula
\begin{equation}
\label{eq:asymp-bessel-eigenvalues}
	\lambda_n = \frac{\pi^2}{b^2}(n + \kappa_\nu)^2 + O(n^{1/2}),
	\quad n\to\infty,
	\quad 
\kappa_\nu\defeq \begin{cases}
					\frac{2\nu+1}{4},& \gamma\ne 0,
					\\
					\frac{2\nu-1}{4},& \gamma  = 0.
				 \end{cases}
\end{equation}
This implies that, contrary to the case of regular Schr\"odinger operators, the 
addition of a point to the spectrum is still compatible with 
\eqref{eq:asymp-bessel-eigenvalues} but for a different value of $\nu$, namely,
for $\tilde{\nu} = \nu + 2$.

\subsection*{Acknowledgments}
We thank A. Poltoratski for helpful comments which led to
Remark~\ref{rem:poltoratsky} and G. Teschl for useful
remarks. J.~H.~T. thanks IIMAS-UNAM for their kind hospitality.

\def\cprime{$'$}

\end{document}